\documentclass[a4paper,11pt]{article}

\usepackage{amsmath, amsthm, amssymb}
\usepackage{tabularx}
\usepackage{algorithm}
\usepackage{algpseudocode}
\usepackage{tikz}
\usetikzlibrary {arrows.meta,automata,positioning}
\usepackage{hyperref}
\hypersetup{breaklinks=True}
\usepackage{xurl}
\usepackage{geometry}

\makeatletter
\newcommand{\multiline}[1]{%
  \begin{tabularx}{\dimexpr\linewidth-\ALG@thistlm}[t]{@{}X@{}}
    #1
  \end{tabularx}
}
\makeatother

\newtheorem{theorem}{Theorem}[section]
\newtheorem{corollary}[theorem]{Corollary}

\newtheorem{lemma}[theorem]{Lemma}

\theoremstyle{definition}
\newtheorem{definition}[theorem]{Definition}
\newtheorem{example}[theorem]{Example}
\newtheorem{remark}[theorem]{Remark}

\title{Construction and Decoding of Convolutional Codes with optimal Column Distances}

\author{Julia Lieb and Michael Schaller}

\begin{document}

\maketitle

\begin{abstract}
The construction of Maximum Distance Profile (MDP) convolutional codes in general requires the use of very large finite fields.
In contrast, convolutional codes with optimal column distances maximize the column distances for given code parameters $(n, k, \delta)$ and a given arbitrary finite field.
In this paper, we present a construction of such convolutional codes.
In addition, we prove that for the considered parameters the codes that we constructed are the only ones achieving optimal column distances.
The structure of the presented convolutional codes with optimal column distances is strongly related to first order Reed-Muller block codes and we leverage this fact to develop a reduced complexity version of the Viterbi algorithm for these codes.
\end{abstract}

\section{Introduction}

Convolutional codes form a class of error-correcting codes which is mainly used for the transmission of streams of data.
In contrast to classical linear block codes, where data blocks of fixed size are encoded and decoded independently, convolutional codes possess memory, i.e., in the encoding process dependencies between different data blocks are created.
These dependencies allow correction of error patterns that cannot be corrected with block codes of the same length \cite{tomas2012decoding}.
The drawback is that this additional security against errors also makes the decoding process more complex.
The complexity of the Viterbi algorithm, the most famous decoding algorithm for convolutional codes, grows exponentially with the memory of the code (see \cite{viterbi}).
Hence, there have been various attempts to make this algorithm more efficient.

In addition to the efficiency of the decoding process, the other crucial property of an error-correcting code is the number of errors it can correct.
An important class of measures for the error-correcting capability of a convolutional code is given by its column distances.
There exist Singleton-like upper bounds for these column distances and convolutional codes that reach a maximal number of these upper bounds with equality are called Maximum Distance Profile (MDP) codes.
While MDP codes have optimal error-correcting properties, their construction requires in general very large finite fields, see e.g.
\cite{GRS:sMDS}, \cite{ANP:MDP}, \cite{rscc}, which makes such constructions often impractical.

Therefore, our approach is to first fix the finite field (as well as the code parameters $(n, k, \delta)$)  arbitrarily and then determine the largest possible column distances over this given finite field, where we prioritize maximizing the first column distances.
In this way we define the notion of optimal column distances.
We then present a construction for convolutional codes with optimal column distances for certain code rates.
Moreover, we prove that for the given code parameters this construction is the only one achieving optimal column distances.
Our main construction which has optimal column distances uses MacDonald codes as building blocks.
Since MacDonald codes can be described in terms of first order Reed-Muller codes, we are able to present a reduced complexity version of the Viterbi algorithm for these codes using a decoder for first order Reed-Muller codes.
Note that this generalizes the work in \cite{isit23} and in \cite{isit24} to arbitrary finite fields.
In \cite{isit23} binary convolutional codes with optimal column distances were constructed and in \cite{isit24} an efficient decoding algorithm for these codes was presented.

To obtain also good constructions for different code rates than the ones covered by our optimal construction, we modify this initial construction using directly first order Reed-Muller codes or simplex codes as building blocks. Doing this, we lose the optimality of the column distances but the column distances are still close to optimal as can be seen from comparison with the column distances of the optimal construction for different values of $n$. In addition, the modified constructions can be decoded with the Viterbi algorithm with the same reduced complexity as the original construction.

The paper is structured as follows.
In Section \ref{sec:preliminaries}, we give some background on block codes and convolutional codes, in Section \ref{sec:viterbi} we describe the Viterbi algorithm for convolutional codes.
In Section \ref{sec:construction}, we present our construction of convolutional codes with optimal column distances.
In Section \ref{sec:decoding}, we explain a decoding algorithm for first order Reed-Muller codes and MacDonald codes that we will use in the following.
In Section \ref{sec:NewDecodingAlgorithm}, we present an improved Viterbi algorithm for the codes constructed in Section \ref{sec:construction} using results of Section \ref{sec:decoding} and analyze the complexity of this algorithm.
In Section \ref{sec:alternative_constructions}, we present two alternative constructions with good column distances and the same reduced complexity with the Viterbi algorithm as the construction of Section \ref{sec:construction}.
Section \ref{sec:conclusion} contains the conclusion and some open problems.

\section{Preliminaries}
\label{sec:preliminaries}

Let $\mathbb F_q$ denote a finite field with $q$ elements, $\mathbb{F}_q^{k \times n}$ the space of $k \times n$ matrices with entries in $\mathbb{F}_q$, and $\mathbb F_q[z]$ the ring of polynomials over $\mathbb F_q$.

\subsection{Linear block codes}
In this subsection, we introduce linear block codes with a focus on MacDonald codes, which we will use later for our convolutional code construction.

\begin{definition}

A $q$-ary $[n,k]$ \textbf{linear block code} $C$ of length $n$ and dimension $k$ is a $k$-dimensional subspace of $\mathbb F_q^n$.
This means that there exists a matrix $G\in \mathbb{F}_q^{k \times n}$ of full row rank such that
$$
C = \left\{u G \in \mathbb F_q^{n}:\, u \in \mathbb F_q^{k}\right\}.$$
We call any such matrix $G$ a \textbf{generator matrix}.
A code $\tilde{C}$ obtained from $C$ by deleting one or several columns of any of its generator matrices is called a \textbf{punctured code} or \textbf{puncturing} of $C$.

For $v \in \mathbb{F}_q^n$ we define the \textbf{Hamming weight} of $v$ as
\[
    \operatorname{wt}(v) = |\{i \in \{1, \ldots, n\} : v_i \neq 0\}|.
\]
For $u, v \in \mathbb{F}_q^n$ we define the \textbf{Hamming distance} between $u$ and $v$ as
\[
    d(u, v) = \operatorname{wt}(u-v).
\]
We define the \textbf{minimum Hamming distance} of $C$ as
$$d_{min}(C)=\min_{u,v\in  C,\ u\neq v}d(u,v)=\min_{c\in C\setminus \{0\}}\operatorname{wt}(c). $$
\end{definition}
We only use the Hamming weight and distance for block codes and therefore we will refer to the above terms simply as weight and distance.

\begin{definition}
    A matrix $M \in \mathbb{F}_q^{n \times n}$ is called a \textbf{monomial matrix} if it can be written as
    \[
        M = DP
    \]
    where $D\in \mathbb{F}_q^{n \times n}$ is an invertible diagonal matrix and $P\in \mathbb{F}_q^{n \times n}$ is a permutation matrix.
\end{definition}

Note that multiplying a generator matrix of a linear code by a monomial matrix on the right gives a code with the same parameters as the original code.

In the following, we define and describe MacDonald codes, the building blocks for our convolutional code construction with optimal column distances.
MacDonald codes were introduced in \cite{macdonald1960design} for the binary case and generalized to arbitrary finite fields in \cite{qaryMacD}. These codes are closely related to simplex codes, which we therefore recall first.

Let $S(q,m)\in\mathbb F_q^{m\times \frac{q^m-1}{q-1}}$ be a generator matrix of a $q$-ary simplex code $\mathcal{S}(q,m)$, i.e., the columns of $S(q,m)$ form a set of representatives for the $\frac{q^m-1}{q-1}$ distinct one-dimensional subspaces of $\mathbb F_q^m$.
Note that row operations do not change the fact that any two columns of a generator matrix are linearly independent.
Therefore, row operations just change the set of column representatives.
This gives the following lemma.
\begin{lemma}
\label{lemma:simplex_code_monomial_right}
    Let $S(q, m), S'(q, m)$ be two generator matrices of simplex codes of the same dimension, i.e., both of whose columns form a set of representatives of the one-dimensional subspaces of $\mathbb{F}_q^m$ (without necessarily being the same columns).
    Then
    \[
        S'(q, m) = S(q, m) M
    \]
    for some monomial matrix $M$.
\end{lemma}
\begin{proof}
    It is clear that any two sets of column representatives are related by right multiplication with a monomial matrix.
\end{proof}
Independently of the choice of representatives, there are exactly $\frac{q^{m-k}-1}{q-1}$ of these columns whose first $k$ entries are equal to zero.
If we puncture a simplex code in exactly these positions, we obtain a \textbf{MacDonald code}, denoted by $\mathcal{C}_{m,m-k}(q)$, which is a $q$-ary code of length $\frac{q^m-q^{m-k}}{q-1}$, dimension $m$ and minimum distance $q^{m-1}-q^{m-k-1}$.
More precisely we have the following lemma.

\begin{lemma}\cite{seneviratne2018generalized}
\label{lemma:weight_Mac_Donald_Codes}
    If we denote by $C_{m,m-k}(q)$ a generator matrix of $\mathcal{C}_{m,m-k}(q)$ that is obtained in the way described above, then all codewords of $\mathcal{C}_{m,m-k}(q)$ obtained as non-zero linear combinations of the first $k$ rows of $C_{m,m-k}(q)$ have weight $q^{m-1}$ and all other codewords have weight $q^{m-1}-q^{m-k-1}$.
\end{lemma}

Note that for $k=1$, MacDonald codes are equal to \textbf{first order Reed-Muller codes} of dimension $m$, denoted by $\mathcal{R}(q,m-1)$.
We denote a generator matrix of $\mathcal{R}(q,m-1)$ by $R(q,m-1)$.
Note that $R(q, m-1)$ can be chosen such that its first row is equal to the all-one vector.
Moreover, for any $k\in\{1,\hdots,m-1\}$, there is a generator matrix of a MacDonald code of the form
\begin{align}\label{eq:MacDonald_RM}
C_{m,m-k}(q)=\left(
    \begin{array}{c|c|c|c|c}
        R(q,m-1)   & 0_{1 \times q^{m-2}}            & 0_{2 \times q^{m-3}} & \ldots & 0_{(k-1)\times q^{m-k}}\\
                                & R(q,m-2)  &  & \ldots &   \\
                                &                       & R(q,m-3) &\ldots &   \\
                          &             &     &\vdots     & R(q,m-k)
    \end{array}
    \right).
\end{align}
Throughout this paper, we will always assume that generator matrices $C_{m,m-k}(q)$ of MacDonald codes are of the form above where all involved generator matrices of first order Reed-Muller codes have the first row equal to the all-one vector.
Note that in
\cite{isit23}, MacDonald codes $\mathcal{C}_{m,m-k}(q)$ were also called
 \textbf{$k$-partial simplex codes} of dimension $m$.

Due to the close relation of MacDonald codes and simplex codes two important properties of simplex codes will be crucial for our convolutional code construction.
Firstly, all non-zero codewords of $\mathcal{S}(q,m)$ have weight $q^{m-1}$, i.e., simplex codes are one-weight codes, and secondly, simplex codes reach equality in the Plotkin bound which we present next.

\begin{theorem}\cite[Theorem 5.2.4]{van1998introduction}
Let $C\subset\mathbb F_q^n$ be a linear block code with $d:=d_{min}(C)>n(1-\frac{1}{q})$. Then,
$$|C|\leq \frac{d}{d-n(1-\frac{1}{q})}.$$
\end{theorem}

The result of the following lemma can be seen as an intermediate step in the proof of the classical Plotkin bound (see e.g. \cite{van1998introduction}).
Even though it is a long known result, we will provide the proof so that the paper is more self-contained.
Later we will use this result for the proof that the convolutional codes we construct have optimal column distances.

\begin{lemma}\label{lemma:plotkin}
For a linear block code $\mathcal{C}\subset\mathbb F_q^n$, one has
$$\sum_{c\in\mathcal{C}}\operatorname{wt}(c)\leq n\cdot  |\mathcal{C}|\cdot \left(1-\frac{1}{q}\right).$$
\end{lemma}

\begin{proof}
We write $\mathbb F_q=\{a_1,\hdots, a_q\}$ and define $m_{i,j}:=|\{c\in\mathcal{C}\ |\ c_i=a_j\}|$ for $i\in\{1,\hdots,n\}$, $j\in\{1,\hdots,q\}$. Note that $\sum_{j=1}^qm_{i,j}=|\mathcal{C}|$ for all $i\in\{1,\hdots,n\}$. Furthermore, we have
$$\sum_{u,v\in\mathcal{C}}d(u,v)=\sum_{i=1}^n|\{u,v\in\mathcal{C}\ |\ u_i\neq v_i\}|=\sum_{i=1}^n\sum_{j=1}^q m_{i,j}(|\mathcal{C}|-m_{i,j}).$$
As in \cite[p. 67]{van1998introduction}, we obtain the upper bound
$$\sum_{j=1}^q m_{i,j}(|\mathcal{C}|-m_{i,j})=|\mathcal{C}|^2-\sum_{j=1}^qm_{i,j}^2\leq |\mathcal{C}|^2\left(1-\frac{1}{q}\right),$$
where in the last step the inequality of arithmetic and quadratic means is used
in the form $\sum_{j=1}^qm_{i,j}^2\geq \frac{1}{q}\left(\sum_{j=1}^q m_{i,j}\right)^2=\frac{1}{q}|\mathcal{C}|^2$.
Moreover, for linear codes,
$$\sum_{u,v\in\mathcal{C}}d(u,v)=\sum_{u,v\in\mathcal{C}}\operatorname{wt}(u-v)=|\mathcal{C}|\sum_{c\in\mathcal{C}}\operatorname{wt}(c)$$
since for each $c,v\in\mathcal{C}$ there exists $u\in\mathcal{C}$ with $c=u-v$ and hence, for each $c\in\mathcal{C}$, there are $|\mathcal{C}|$ ways to write it as difference of two codewords.
Finally, this yields
$$|\mathcal{C}|\sum_{c\in\mathcal{C}}\operatorname{wt}(c)\leq n \cdot  |\mathcal{C}|^2\cdot \left(1-\frac{1}{q}\right),$$
and dividing by $|\mathcal{C}|$ proves the lemma.
\end{proof}

One can see that simplex codes $\mathcal{S}(q,m)$ achieve equality in the bound of Lemma \ref{lemma:plotkin} as both sides of the inequality are equal to $q^{m-1}(q^m-1)$.

\subsection{Convolutional codes}
In this subsection we introduce convolutional codes and some of their basic properties.
For more details and further references you can consult for example \cite{jo99}.

\begin{definition}
An $(n,k)$ \textbf{convolutional code} $\mathcal{C}$ is defined as an $\mathbb{F}_{q}[z]$-submodule of $\mathbb{F}_{q}[z]^n$ of rank $k$. A matrix $G(z)\in \mathbb{F}_{q}[z]^{k \times n}$ whose rows form a basis of $\mathcal{C}$ is called a \textbf{generator matrix} for $\mathcal{C}$, i.e.,
{\begin{eqnarray}
\mathcal{C} = \{{u(z)G(z) \in \mathbb{F}_{q}[z]^{n}: u(z) \in \mathbb{F}_{q}[z]^{k}\}.}\nonumber
\end{eqnarray}}
\end{definition}

Linear block codes can be seen as convolutional codes with constant generator matrix. There are several notions of degree related to a convolutional code, which we will introduce in the following.

\begin{definition}
Let $G(z)=\sum_{i=0}^{\mu}G_i z^{i}\in\mathbb F_q[z]^{k\times n}$ with $G_{\mu} \neq 0$ and $k\leq n$.
\begin{itemize} \setlength{\itemsep}{5pt}
    \item For each $i\in\{1,\hdots,k\}$, the $i$-th \textbf{row degree} $\nu_i$ of $G(z)$ is defined as the largest degree of any entry in row $i$ of $G(z)$, and the largest row degree $\mu=\max_{i=1,\dots,k}\nu_i$ is called the \textbf{memory} of $G(z)$.
    \item The \textbf{external degree} of $G(z)$ is defined as the sum of the row degrees of $G(z)$.
    \item The \textbf{internal degree} of $G(z)$ is defined as the maximal degree of the  $k\times k$ minors of $G(z)$.
\end{itemize}

\end{definition}

\begin{definition}
\begin{itemize} \setlength{\itemsep}{5pt}
    \item A matrix $G(z)\in \mathbb F_q[z]^{k\times n}$ is said to be \textbf{row reduced}\index{row reduced} if its internal and external degrees are equal.
    In this case, $G(z)$ is called a \textbf{minimal} generator matrix of the convolutional code it generates.
    \item The \textbf{degree} $\delta$ of a code $\mathcal{C}$ is defined as the external degree of a minimal generator matrix of $\mathcal{C}$.
    \item An $(n,k)$ convolutional code with degree $\delta$ is called an $(n,k,\delta)$ convolutional code.
    \item A minimal generator matrix $G(z)$ is said to have \textbf{generic row degrees} if
$\mu=\left\lceil\frac{\delta}{k} \right\rceil$ and exactly $k\left\lceil\frac{\delta}{k}\right\rceil-\delta$ rows of $G_{\mu}$ are zero.
\item We denote by $\Tilde{G}_{\mu}\in\mathbb F_q^{(\delta+k-k\left\lceil\frac{\delta}{k}\right\rceil)\times n}$ the matrix consisting of the $\delta+k-k\left\lceil\frac{\delta}{k}\right\rceil$ non-zero rows of  $G_{\mu}$.
In addition, we denote by $\tilde{u}_i\in\mathbb F_q^{\delta+k-k\left\lceil\frac{\delta}{k}\right\rceil}$ the vector consisting of the corresponding $\delta+k-k\left\lceil\frac{\delta}{k}\right\rceil$ components of $u_i\in\mathbb F_q^k$.
\end{itemize}

\end{definition}

\begin{definition}
    Let $\mathcal{C}$ be an $(n, k, \delta)$ convolutional code.
    Then we say that $\mathcal{C}$ admits a \textbf{parity-check matrix}, if there exists $H(z)\in\mathbb F_q[z]^{(n-k)\times n}$ of full rank such that
  $$        \mathcal{C} = \{ v(z) \in \mathbb F_q[z]^n \ |\ H(z)v(z)^\top = 0 \}.$$
  If such a parity-check matrix exists for $\mathcal{C}$, then $\mathcal{C}$ is called \textbf{non-catastrophic}.
\end{definition}

Not every convolutional code admits a representation via a parity-check matrix.
The following theorem gives a characterization of codes that admit a parity-check matrix.

\begin{theorem}\cite{von1997algebraic}
    Let $\mathcal{C}$ be an $(n, k, \delta)$ convolutional code with generator matrix $G(z)\in\mathbb F_q[z]^{k\times n}$.
    Then, $\mathcal{C}$ admits a parity-check matrix if and only if $G(z)$ is left prime, i.e., if
    \[
        \operatorname{rk}( G(\lambda) )= k
    \]
     for all $\lambda \in \overline{\mathbb F}_q$, where $\overline{\mathbb F}_q$ denotes the algebraic closure of $\mathbb F_q$.
\end{theorem}

Note that if $G(z),\tilde{G}(z)\in\mathbb F_q[z]^{k\times n}$ are generator matrices of the same convolutional code, then $G(z)=U(z)\tilde{G}(z)$ for some $U(z)\in\mathbb F_q[z]^{k\times k}$ with $\det(U(z)) \in \mathbb{F}^\ast_q$. Hence, $G(z)$ is left prime if and only if $\tilde{G}(z)$ is left prime.

Codewords $v(z)=\sum_{i\in\mathbb{Z}_{\geq 0}}v_iz^i$ of convolutional codes are often understood as sequences of coefficient vectors $(v_i)_{i\in\mathbb{Z}_{\geq 0}}$, where $i$ represents the time unit at which $v_i$ is sent and received (possibly with errors). If the coefficient vector $u_i$ of a message $u(z)=\sum_{i\in\mathbb{Z}_{\geq 0}}u_iz^i$ can be recovered with the knowledge of $v_0,\hdots,v_{i+j}$, one says that it can be recovered with time delay $j$.

\begin{definition}
A generator matrix $G(z)\in {\mathbb{F}_{q}[z]}^{k\times n}$ with $G_0=G(0)$ of full row rank is called \textbf{delay-free}.
\end{definition}

Note that all non-catastrophic convolutional codes are delay-free.
Furthermore, note that if a convolutional code is not delay-free, then it is not possible to uniquely recover a message from its corresponding codeword without delay, even if no errors occurred during transmission.
The reason is as follows.
We have $v_i=u_iG_0+\sum_{\ell=0}^{i-1}u_{\ell}G_{i-\ell}$.
Hence, in order to recover $u_i$ from $v_i$ with the knowledge of \break $u_0,\hdots,u_{i-1}$, it is needed that $G_0$ has full rank.
Therefore, we will only work with delay-free convolutional codes in this paper.

Similarly as for block codes there exist notions of distance to measure the error-correcting capability of convolutional codes.
In the following, we present the two most important distance notions, the free distance and the column distance.

\begin{definition}
The Hamming weight
of $v(z)=\sum_{t\in\mathbb{Z}_{\geq 0}}v_tz^t \in \mathbb{F}_q[z]^n$ is defined as $\operatorname{wt}(v(z))=\sum_{t\in\mathbb{Z}_{\geq 0}}\operatorname{wt}(v_t)$, where $\operatorname{wt}(v_t)$ is the Hamming weight of $v_t\in\mathbb F_q^n$.
\end{definition}

\begin{definition}
{\em
The \textbf{free distance}\index{code!minimum distance}\index{distance!free} of a convolutional code
  $\mathcal{C}$ is given by
  \[d_{\mathrm{free}}(\mathcal{C}):=\min_{v(z)\in\mathcal{C}}\left\{\operatorname{wt}(v(z))\
    |\ v(z) \neq 0\right\}.\]
    }
\end{definition}

\begin{definition}
Let $G(z)=\sum_{i=0}^{\mu}G_iz^i\in \mathbb{F}_q[z]^{k\times n}$ be a generator matrix of a convolutional code $\mathcal{C}$. For $j\in\mathbb{Z}_{\geq 0}$, we define the \textbf{truncated sliding
generator matrices} as
\begin{align*}
  G_j^c:=\arraycolsep=5pt
  \left(
  \begin{array}{ccc}
    G_0 &    \cdots & G_j \\
         &  \ddots  &  \vdots  \\
       &  &  G_0 
  \end{array}
              \right)\in \mathbb F_q^{(j+1)k\times (j+1)n}
              \end{align*}
where we set $G_i=0$ for $i>\mu$.
\end{definition}

\begin{definition}

For $j\in\mathbb{Z}_{\geq 0}$, the \textbf{j-th column distance}\index{column distance} of a convolutional code $\mathcal{C}$ is defined as
\[
d_j^c(\mathcal{C}):=\min\left\{\operatorname{wt}(v_0,\hdots,v_j)\ |\ {v}(z)\in\mathcal{C} \text{ and }{v}_0 \neq 0\right\}.
\]
\end{definition}
Since the convolutional codes which we will construct in this paper will all be delay-free, we can use that in this case
\begin{align*}
d_j^c(\mathcal{C})&=\min\left\{\operatorname{wt}(u_0,\hdots,u_j)G_j^c\ |\
{u}_0 \neq 0\right\}\\
&= \min_{u_0\neq 0}\sum_{i=0}^j \operatorname{wt}\left((u_0, \ldots, u_i)\left(\begin{array}{c}
    G_i\\ \vdots \\ G_0 \nonumber
\end{array}\right)\right).
\end{align*}
Note that the $j$-th column distance is a measure for the number of errors that can be corrected with time delay $j$ and hence, the column distances are the crucial measure for low-delay applications.
The next theorem describes the relation between the free distance and the column distances.

\begin{theorem}\cite{jo99}
For an $(n,k,\delta)$ convolutional code $\mathcal{C}$, one has
 $$0\leq d_0^c(\mathcal{C})\leq d_1^c(\mathcal{C})\leq\cdots\leq d_{\mathrm{free}}(\mathcal{C})$$ and if $\mathcal{C}$ is non-catastrophic, then
$\displaystyle d_{\mathrm{free}}(\mathcal{C})= \lim_{j\rightarrow\infty} d_j^c (\mathcal{C})$.
\end{theorem}

In the following, we give upper bounds for the column distances and consider convolutional codes with maximal column distances.

\begin{theorem}\cite[Proposition 2.2, Corollary 2.3]{GRS:sMDS},\cite[Theorem 4]{lieb2025new}
Let $\mathcal{C}$ be an $(n,k,\delta)$ delay-free convolutional code.
For $j \in \mathbb{Z}_{\geq 0}$,
$$
d^c_j ({\mathcal{C}})\leq (n-k)(j+1)+1.
$$
Moreover, if $d^c_j ({\mathcal{C}})= (n-k)(j+1)+1$ for some $j \in \mathbb{Z}_{\geq 0}$, then $d^c_i({\mathcal{C}}) = (n-k)(i+1)+1$, for $i \leq j$.
\end{theorem}

\begin{definition}
\label{def:mdp}
An $(n,k,\delta)$ delay-free convolutional code with
$$d^c_j(\mathcal{C})= (n-k)(j+1)+1\quad\text{for}\quad j \leq L:=\left\lfloor\frac{\delta}{k}\right\rfloor+\left\lfloor\frac{\delta}{n-k}\right\rfloor$$ is called a \textbf{Maximum Distance Profile (MDP)} convolutional code.
\end{definition}

Since MDP codes are in general hard to construct and their construction requires very large finite fields (see e.g. \cite{GRS:sMDS, ANP:MDP, rscc}), our aim is to construct convolutional codes that are optimal in the sense of the following definition.
Having in mind that the $j$-th column distance is the measure for the number of errors that can be corrected with time delay $j$, we want to first maximize the column distances for small values of $j$.

\begin{definition}\label{defopt}
   We say that a delay-free $(n,k,\delta)$ convolutional code $\mathcal{C}\subset\mathbb F_q[z]^n$ has \textbf{optimal column distances} if there exists no delay-free $(n,k,\delta)$ convolutional code $\hat{\mathcal{C}}\subset\mathbb F_q[z]^n$ such that $d^c_j(\hat{\mathcal{C}})>d^c_j(\mathcal{C})$ for some $j\in\mathbb{Z}_{\geq 0}$ and $d^c_i(\hat{\mathcal{C}})=d^c_i(\mathcal{C})$ for all $0\leq i<j$. Such convolutional codes have also been called \textbf{Optimum Distance Profile (ODP)} convolutional codes in the literature before; see e.g. \cite{cedervall2002fast}.
\end{definition}

\section{The Viterbi algorithm for convolutional codes}
\label{sec:viterbi}

In this section, we will describe the general idea of the Viterbi algorithm.
For more information see \cite{viterbi}.
For ease of notation and since we will later use the Viterbi algorithm (in an improved version) only for row reduced generator matrices with generic row degrees, we will present the Viterbi algorithm only for such generator matrices.

A convolutional encoder can be represented by a  trellis diagram as in Figure \ref{trellis}, with the \textbf{states} $S_t=(u_{t-1},\dots, \tilde{u}_{t-\mu})\in\mathbb F_q^{\delta}$ (represented with circles in the diagram) at each time unit $t$ representing the previous $\mu$ inputs.
We do not consider the components of $u_{t-\mu}$ which are multiplied with zero rows of $G_{\mu}$ when calculating $c_t$.
Transitions (represented by arrows in the diagram) occur only between states $S_t$ at time unit $t$ and states $S_{t+1}$ at time unit $t + 1$.
Each path in the trellis corresponds to a codeword $c(z)=u(z)G(z)$ and each arrow for a transition from $t$ to $t+1$ is labeled with $u_t/c_t$.

The Viterbi algorithm for convolutional codes can be seen as ``dynamic programming'' applied to a trellis diagram \cite{forney2005}.
In the following we will present this algorithm and illustrate it with an example.

Inputs for the algorithm are a row reduced generator matrix $G(z)$ with memory $\mu$ and generic row degrees and the received polynomial $r(z) = \sum_{i=0}^{N-1}r_iz^i \in \mathbb{F}_q[z]^n$ with $r_i \in \mathbb{F}_q^n$.
We call $N$ the \textbf{length} of the sequence $(r_0,\hdots,r_{N-1})$.
If we have a message $u(z)$, we will send $c(z) = u(z)G(z)$ and assume we know that $\deg(c(z))=N-1.$
This condition will be used as a stopping criterion for the Viterbi algorithm since $$\deg(u(z)) \leq N-1 - \mu
.$$
The algorithm is as follows:

    For $t\in\{1,\hdots,N-\mu\}$, for every $u_0,\hdots, u_{t-1}\in \mathbb{F}_q^k$, compute $$c_{t-1} = \sum_{j=0}^{t-1} u_j G_{t-1-j} \quad \text{and} \quad \sum_{j=0}^{t-1} d(c_j, r_j).$$
    If the state $(u_{t-1}, \dots , \tilde{u}_{t-\mu})$ appears several times, which can happen only if $t\geq \mu+1$, we only keep (one of) the corresponding codewords with the smallest distance $\sum_{j=0}^{t-1} d(c_j, r_j)$.
    Note that for $t< \mu$, one reaches $q^{tk}$ different states and for $t\geq\mu$ one reaches $q^{\delta}$ different states.

    For $t\in\{N-\mu+1, \dots , N\}$, which correspond to the termination steps of the algorithm, we do the same but with fixing $$u_{N-\mu} = \dots = u_{N-1} = 0.$$
    This implies that for $t\in\{N-\mu+1, \dots , N-1\}$, the number of possible states is equal to $q^{\delta-(t-N+\mu)k}$ and for $t = N$ it is equal to $1$. Therefore, after step $N$, we obtain a unique codeword as the result of the decoding.

\begin{figure}
    \centering
    \tikzstyle{state}=[shape=circle,draw=black,scale=0.7]
\tikzstyle{lightedge}=[<-, dashed,>=stealth, scale=0.7]
\tikzstyle{mainstate}=[state,very thick]
\tikzstyle{mainedge}=[<-,very thick,>=stealth, scale=0.7]
\begin{tikzpicture}[]
\node               at (0,4.3) {$\scriptstyle{t=0}$};
\node[mainstate] (s1_1) at (0,3.8) {$00$};
\node               at (2,4.3) {$\scriptstyle{ t=1}$};
\node[state] (s1_2) at (2,3.8) {$00$}
    edge[lightedge] (s1_1);
\node[mainstate] (s3_2) at (2,2.4) {$10$}
     edge[mainedge] node[midway, above, yshift = -1mm, sloped] {\bf{\tiny{1/11}}}(s1_1);
\node               at (4,4.3) {$\scriptstyle{ t=2}$};
\node[state] (s1_3) at (4,3.8) {$00$}
    edge[lightedge]  (s1_2);
\node[mainstate] (s2_3) at (4,3.1) {$01$}
    edge[mainedge] node[midway, above, xshift=-1mm, yshift = -1mm, sloped] {\bf{\tiny{0/01}}}(s3_2);
\node[state] (s3_3) at (4,2.4) {$10$}
    edge[lightedge] (s1_2);
\node[state] (s4_3) at (4,1.7) {$11$}
    edge[lightedge] (s3_2);
\node               at (6,4.3) {$\scriptstyle{ t=3}$};
\node[state] (s1_4) at (6,3.8) {$00$}
    edge[lightedge]  (s1_3)
    edge[lightedge]  (s2_3);
\node[state] (s2_4) at (6,3.1) {$01$}
    edge[lightedge] (s3_3)
    edge[lightedge] (s4_3);
\node[mainstate] (s3_4) at (6,2.4) {$10$}
    edge[lightedge] (s1_3)
    edge[mainedge] node[midway, above, yshift = -1mm, xshift=-1mm, sloped] {\bf{\tiny{1/00}}}(s2_3);
\node[state] (s4_4) at (6,1.7) {$11$}
    edge[lightedge] (s3_3)
    edge[lightedge] (s4_3);
\node               at (8,4.3) {$\scriptstyle{ t=4}$};
\node[state] (s1_5) at (8,3.8) {$00$}
    edge[lightedge]  (s1_4)
    edge[lightedge]  (s2_4);
\node[state] (s2_5) at (8,3.1) {$01$}
   edge[mainedge] node[midway, above, yshift = -1mm, sloped] {\bf{\tiny{0/01}}} (s3_4)
    edge[lightedge] (s4_4);
\node               at (10,4.3) {$\scriptstyle{ t=5}$};
\node[state] (s1_6) at (10,3.8) {$00$}
    edge[lightedge]  (s1_5)
    edge[mainedge] node[midway, above, yshift = -1mm, xshift=-1mm, sloped] {\bf{\tiny{0/11}}}  (s2_5);
\end{tikzpicture}
    \caption{Trellis for the code of Example \ref{example:viterbi} with the sent codeword highlighted}
    \label{trellis}
\end{figure}
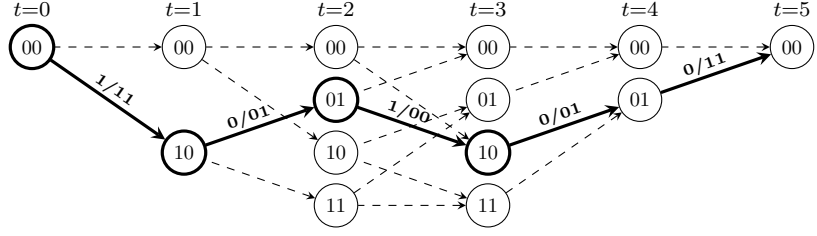

\begin{example}
\label{example:viterbi}

    Consider the $(2, 1, 2)$ convolutional code $\mathcal{C}$ with generator matrix $G(z) = \begin{pmatrix}1+z^2 & 1+z+z^2\end{pmatrix} \in \mathbb{F}_2[z]^{1\times 2}$, i.e., the code has memory $\mu = 2$ and
    \begin{align*}
        G_0 = \begin{pmatrix}1 & 1\end{pmatrix}, \quad G_1 = \begin{pmatrix}0 & 1\end{pmatrix}, \quad G_2 = \begin{pmatrix}1 & 1\end{pmatrix}.
    \end{align*}
    Assume we received $r(z) = \begin{pmatrix}1+z^3+z^4  & 1+z+z^3+z^4\end{pmatrix}$, i.e., $N=\deg(r(z))+1=5$ and
    \begin{align*}
        \begin{array}{c|c|c|c|c}
            r_0 &  r_1 & r_2 & r_3 & r_4  \\
            \hline
            1 \; 1 & 0\;1 & 0\;0 & 1\;1 & 1\;1
        \end{array}
    \end{align*}
    We get $N-1 -\mu = 4-2 = 2$.
    The possible states are given by all tuples of the form $(u_{t-1}, u_{t-2}) \in \mathbb{F}_2^2$, i.e.,
    $        \begin{pmatrix}0 \ 0\end{pmatrix},
        \begin{pmatrix}0 \ 1\end{pmatrix},
        \begin{pmatrix}1 \ 0\end{pmatrix},
        \begin{pmatrix}1 \ 1\end{pmatrix}
    $.

    We go through the algorithm step by step for $t = 1, \dots, N=5$.
    \begin{itemize} \setlength{\itemsep}{6pt}
        \item $t = 1$:
            For all $u_0 \in \mathbb{F}_2$ compute $c_0 = u_0G_0$ and $d(c_0, r_0)$.
            \begin{itemize}
                \item $u_0 = 0:$ $c_0 = (0 \; 0)$ and $d( (0\; 0) , (1\; 1)) = 2$
                \item $u_0 = 1:$ $c_0 = (1 \; 1)$ and $d( (1\; 1) , (1\; 1)) = 0$
            \end{itemize}
        \item $t = 2$:
            For all $u_0, u_1 \in \mathbb{F}_2$ compute $c_1 = u_0 G_1 + u_1 G_0$ and \\$d=d((c_0\ c_1), (r_0\ r_1)) = d(c_0, r_0) + d(c_1, r_1)$.
             \begin{itemize}
                \item $(u_1\ u_0) = (0 \ 0):$ $c_1 = (0\; 0)$ and $d = 2 + 1 = 3$
                \item $(u_1\ u_0) = (0 \ 1):$ $c_1 = (0\; 1)$ and $d = 0 + 0 = 0$
                \item $(u_1\ u_0) = (1 \ 0):$
                $c_1 = (1\; 1)$ and $d = 2 + 1 = 3$
                \item $(u_1\ u_0) = (1 \ 1):$ $c_1 = (1\; 0)$ and $d = 0 + 2 = 2$
            \end{itemize}
        \item $t = 3$:
            For all $u_0, u_1, u_2 \in \mathbb{F}_2$ compute $c_2 = u_0 G_2 + u_1 G_1  +u_2 G_0$ and $d=d((c_0\ c_1\ c_2), (r_0\ r_1\ r_2)) = d(c_0, r_0) + d(c_1, r_1)+ d(c_2, r_2)$.
            The first two components of $(u_2\ u_1\ u_0)$ denote now the state.
            \begin{itemize} \setlength{\itemsep}{5pt}
                \item $(u_2\ u_1\ u_0)=(0\ {0\ 0}):$ $c_2 = (0 \; 0)$ and $d = 3 + 0 = 3$
                \item $(u_2\ u_1\ u_0)=(0\ {0\ 1}):$ $c_2 = (1 \; 1)$ and $d = 0 + 2 = \mathbf{2}$
                \item $(u_2\ u_1\ u_0)=(0\ {1\ 0}):$ $c_2 = (0 \; 1)$ and $d = 3 + 1 = 4$
                  \item $(u_2\ u_1\ u_0)=(0\ {1\ 1}):$ $c_2 = (1 \; 0)$ and $d = 2 + 1 = \mathbf{3}$
                \item $(u_2\ u_1\ u_0)=(1\ {0\ 0}):$ $c_2 = (1 \; 1)$ and $d = 3 + 2 = 5$     \item $(u_2\ u_1\ u_0)=(1\ {0\ 1}):$ $c_2 = (0 \; 0)$ and $d = 0 + 0 = \mathbf{0}$
               \item $(u_2\ u_1\ u_0)=(1\ {1\ 0}):$ $c_2 = (1 \; 0)$ and $d = 3 + 1 = 4$
                \item $(u_2\ u_1\ u_0)=(1\ {1\ 1}):$ $c_2 = (0 \; 1)$ and $d = 2+ 1 = \mathbf{3} $
            \end{itemize}
           Each state appears twice and for each state we only keep the path with the smallest distance.
           We marked these smallest values per state, i.e., the values corresponding to the surviving paths, in bold.
        \item For $t = 4=N-\mu+1$, we set $u_3 = 0$ (termination)  and have $c_3 = u_1G_2 + u_2G_1$ and proceed as before to obtain:
            \begin{itemize} \setlength{\itemsep}{6pt}
                \item $(u_3\ u_2\ u_1\ u_0)=(0\ 0 \ 0 \ 1):$ $c_3= (0 \; 0)$ and $d = 2 + 2 = 4$
                \item $(u_3\ u_2\ u_1\ u_0)=(0\ 0 \ 1 \ 1):$ $c_3 = (1\; 1)$ and $d = 3 + 0 = \mathbf{3}$
                \item $(u_3\ u_2\ u_1\ u_0)=(0\ 1 \ 0 \ 1):$ $c_3 = (0\; 1)$ and $d = 0 + 1 = \mathbf{1}$
                \item $(u_3\ u_2\ u_1\ u_0)=(0\ 1 \ 1 \ 1):$ $c_3 = (1\; 0)$ and $d = 3 + 1 = 4$
            \end{itemize}
            Again, for each state, we only keep the path with the smallest distance marked in bold.
        \item  For $t = 5$, we set $u_4 = 0$ (termination) and obtain:
            \begin{itemize} \setlength{\itemsep}{6pt}
               \item $(u_4\ u_3\ u_2\ u_1\ u_0)=(0\ 0 \ 0 \ 1 \ 1)$: $c_4 = (0\; 0)$ and $d = 3 + 2 = 5$
                \item $(u_4\ u_3\ u_2\ u_1\ u_0)=(0\ 0 \ 1 \ 0 \ 1)$: $c_4 = (1\; 1)$ and $d = 1 + 0 = 1$
                \end{itemize}
            Hence, the final surviving path is $(0,\; 0, \; 1, \; 0, \; 1)$.
    \end{itemize}
    \textbf{Decoding result:}  $$u(z) = \sum_{i=0}^4 u_i z^i = 1 + z^2.$$
    \textbf{Corresponding codeword:} $c(z) = u(z)G(z)=\begin{pmatrix}1+z^4 & 1 + z + z^3 + z^4\end{pmatrix}$ which differs from $r(z)$ in one position, namely in the first entry of the coefficient of $z^3$:
    \begin{align*}
        \begin{array}{c|c|c|c|c}
            r_0 &  r_1 & r_2 & r_3 & r_4  \\
            \hline
            1 \; 1 & 0\;1 & 0\;0 & 1\;1 & 1\;1
        \end{array} \qquad
        \begin{array}{c|c|c|c|c}
            c_0 &  c_1 & c_2 & c_3 & c_4  \\
            \hline
            1 \; 1 & 0\;1 & 0\;0 & 0\;1 & 1\;1
        \end{array}
    \end{align*}
\end{example}

\section{Construction of convolutional codes with optimal column distances}
\label{sec:construction}

In this section, we fix an arbitrary finite field $\mathbb F_q$, and present a construction of convolutional codes over $\mathbb F_q$ with optimal column distances;
this generalizes the construction for $q = 2$ from \cite{isit23}.

For our construction we choose the field size arbitrarily, as well as the dimension and degree of
the code.
The length of the code is then determined by the remaining parameters.

We will utilize generator matrices of $q$-ary MacDonald block codes $\mathcal{C}_{\delta+k,\delta}(q)$ for the construction of $(n,k,\delta)$ convolutional codes with optimal column distances.
For ease of notation, we choose the representatives for the columns in the generator matrices of $\mathcal{S}(q,\delta+k)$ and $\mathcal{C}_{\delta+k,\delta}(q)$ such that the first non-zero entry of each column is equal to $1$.
This implies that (for fixed $q$, $\delta$ and $k$) all generator matrices $S(q,\delta+k)$ (and also $C_{\delta+k,\delta}$) differ only by column permutations.
Via suitable column permutations one can achieve that
\begin{align*}
   S(q,\delta+k)=\begin{pmatrix}
    S(q,k) & \cdots & S(q,k) & 0_{k \times \frac{q^{\delta}-1}{q-1}}\\
   X_1 & \cdots & X_{q^{\delta}} & S(q,\delta)
\end{pmatrix}
\end{align*}
where $X_i=(x_i, \ldots, x_i)\in\mathbb F_q^{\delta\times \frac{q^k-1}{q-1}}$ with $x_i\in\mathbb F_q^{\delta}$ and $\{x_1,\hdots,x_{q^{\delta}}\}=\mathbb F_q^{\delta}$.
Therefore we have
\begin{align*}
C_{\delta+k,\delta}(q) =
    \begin{pmatrix}
    S(q,k) & \cdots & S(q,k)\\
   X_1 & \cdots & X_{q^{\delta}}
\end{pmatrix}.
\end{align*}

\begin{example}\label{ex:simplex}
Choosing $q=3$, $\delta=1$ and $k=2$, one obtains
    $$S(3,3)=\left(\begin{array}{cccccccccccccccc}
      1 & 1 & 1 & 0 & \vline & 1 & 1 & 1 & 0 & \vline & 1 & 1 & 1 & 0 & \vline & 0\\
      0 & 1 & 2 & 1 & \vline & 0 & 1 & 2 & 1 & \vline & 0 & 1 & 2 & 1  & \vline & 0\\
      \hline
    0 & 0 & 0 & 0 & \vline & 1 & 1 & 1 & 1 & \vline & 2 & 2 & 2 & 2 & \vline & 1
  \end{array}\right).$$
\end{example}

We are now ready to present our main construction of convolutional codes and prove that it has optimal column distances.
We call it \textbf{Construction 1} in the following.

\begin{theorem}[Construction 1]
\label{thm:construction_1}
The $\left(\frac{q^{\delta}(q^k-1)}{q-1},k,\delta\right)$ convolutional code $\mathcal{C}$ over $\mathbb F_q$ with generator matrix $G(z)=\sum_{i=0}^{\mu}G_iz^i$ where $\mu=\left\lceil\frac{\delta}{k}\right\rceil$ and

$$\left(\begin{array}{c}
    G_0\\ \vdots\\
    \tilde{G}_{\mu}
    \end{array}\right)=
    \begin{pmatrix}
    S(q,k) & \cdots & S(q,k) \\
   X_1 & \cdots & X_{q^{\delta}}
\end{pmatrix}=C_{\delta+k,\delta}(q)$$
is non-catastrophic and has optimal column distances

$$d_j^c(\mathcal{C})=\begin{cases}
q^{\delta+k-1}+j\cdot (q^{\delta+k-1}-q^{\delta-1}) & \text{for}\quad j\leq \left\lfloor\frac{\delta}{k}\right\rfloor,\\
q^{\delta+k-1}+\left\lfloor\frac{\delta}{k}\right\rfloor\cdot (q^{\delta+k-1}-q^{\delta-1})  & \text{for}\quad j\geq \left\lfloor\frac{\delta}{k}\right\rfloor.
\end{cases}$$
Moreover,
$$d_{\mathrm{free}}(\mathcal{C})=q^{\delta+k-1}+\left\lfloor\frac{\delta}{k}\right\rfloor\cdot (q^{\delta+k-1}-q^{\delta-1}).$$

\end{theorem}

\begin{proof}
We start by showing that $\mathcal{C}$ is non-catastrophic.
Recall that by $R(q, m-1)$ we denote a generator matrix of a first order Reed-Muller code of dimension $m$ whose first row is the all ones vector.
As there exist $q^{\delta + k -1}$ columns of
\[
   \left(\begin{array}{c}
    G_0\\ \vdots\\
    \tilde{G}_{\mu}
    \end{array}\right)
\]
which form $R(q,\delta + k -1)$, one has that $(G_0^\top\ \cdots\ \ G_{\mu-1}^\top \ \tilde{G}_{\mu}^\top)^\top$ contains
$$ \left(\begin{array}{ccc}
    1 & \cdots & 1\\
    \hline
    0 & & \\
    \vdots & I_{k-1}\\
    0\\
    \hline
    & 0_{\delta\times k}
\end{array}\right)$$
as submatrix.
This implies that $G(z)$ contains
$$\left(\begin{array}{ccc}
    1 & \cdots & 1\\
    \hline
    0 & & \\
    \vdots & I_{k-1}\\
    0
\end{array}\right)$$ as submatrix, which implies that $G(z)$ has full row rank for all $z\in\overline{\mathbb F}_q$. Hence, $G(z)$ is left prime and $\mathcal{C}$ non-catastrophic.

To show optimality of the column distances, denote by $G_j$ the coefficient matrices of the generator matrix as defined in the theorem and denote by $\hat{G}_j$ the coefficient matrices of a generator matrix of any convolutional code $\mathcal{\hat{C}}$ with the same parameters as in the theorem and with optimal column distances. We prove by induction with respect to $j\in\mathbb{Z}_{\geq 0}$ that the choice of $G_j$ leads to optimal $d_j^c$. Hereby note that $G_0$ and $\hat{G}_0$ have full rank since all considered codes are delay-free.

For $j=0$, we use that $d_0^c(\mathcal{C})$ is the same as the distance of the block code $\mathcal{C}_0$ with generator matrix $G_0$ and $d_0^c(\mathcal{\hat{C}})$ is the same as the distance of the block code $\mathcal{\hat{C}}_0$ with generator matrix $\hat{G}_0$.

We have
 $G_0=\begin{pmatrix}
    S(q,k) \ \cdots \ S(q,k)
\end{pmatrix}$ and all non-zero codewords of the code with generator matrix $G_0$ have weight $q^{k-1} q^{\delta}$.
Hence, $d_0^c(\mathcal{C})=q^{k-1} q^{\delta}\leq d_0^c(\mathcal{\hat{C}})$, by the assumption that $\mathcal{\hat{C}}$ is optimal.
According to Lemma \ref{lemma:plotkin}, we have
\begin{align}\label{eq:plotkin}
    \sum_{c\in\mathcal{\hat{C}}_0}\operatorname{wt}(c)\leq |\mathcal{\hat{C}}_0| n \left(1 - \frac{1}{q}\right)=\sum_{c\in\mathcal{C}_0}\operatorname{wt}(c)
\end{align}
with $n=\frac{q^{\delta}(q^k-1)}{q-1}$.
The last equality is obtained by easy calculations which show that the bound of Lemma \ref{lemma:plotkin} is sharp for $\mathcal{C}_0$.
Since, by assumption, $d_0^c(\hat{C})=d(\hat{C}_0)$ is maximal among all codes with the same parameters, one obtains
\[
    \operatorname{wt}(\hat{c})\geq d(\hat{C}_0)\geq q^{k-1} q^{\delta}=\operatorname{wt}(c)
\]
for all $\hat{c}\in\mathcal{\hat{C}}_0\setminus\{0\}$, $c\in\mathcal{C}_0\setminus\{0\}$.
Thus, together with \eqref{eq:plotkin}, we get
\[
    \operatorname{wt}(\hat{c})=\operatorname{wt}(c)=q^{k-1} q^{\delta}
\]
for all $\hat{c}\in\mathcal{\hat{C}}_0\setminus\{0\}$, $c\in\mathcal{C}_0\setminus\{0\}$.
Therefore, $d_0^c(\mathcal{C})=d_0^c(\mathcal{\hat{C}})$, implying that our construction has optimal $0$-th column distance, and that the block codes with generator matrices
 $G_0$ and $\hat{G}_0$ have the same weight distribution (both are one-weight codes, i.e., all non-zero codewords have the same weight).

Since $G_0$ and $\hat{G}_0$ generate one-weight codes, one has
\begin{align*}
d_1^c(\mathcal{C})&=d_0^c(\mathcal{C})+\min\left\{\operatorname{wt}\left((u_1,  u_0)\begin{pmatrix}
    G_0\\G_1
\end{pmatrix}\right), u_0\neq 0\right\},\\
d_1^c(\mathcal{\hat{C}})&=d_0^c(\mathcal{\hat{C}})+\min\left\{\operatorname{wt}\left((u_1,  u_0)\begin{pmatrix}
    \hat{G}_0\\\hat{G}_1
\end{pmatrix}\right), u_0\neq 0\right\}.
\end{align*}
To be able to do a similar reasoning for $d_1^c$ as we did for $d_0^c$, we append columns to the matrix $\begin{pmatrix}
    G_0\\G_1
\end{pmatrix}$ in such a way that we obtain multiple copies of the generator matrix of a suitable simplex code.
More precisely, one can reorder the columns of the matrix
\begin{equation}
\label{eq:add_simplex_codes}
    \begin{pmatrix}
    G_0 & 0_{k \times \frac{q^k-1}{q-1}} & \cdots & 0_{k \times \frac{q^k-1}{q-1}}\\G_1 & S(q,k) & \cdots & S(q,k)
    \end{pmatrix}
\end{equation}
such that one obtains the matrix
\begin{equation}
\label{eq:reorder_added_simplex_codes}
    \begin{pmatrix}
        S(q,2k) & \cdots & S(q,2k)
    \end{pmatrix}
\end{equation}
where we have $q^{\delta-k}$ copies of the corresponding simplex codes $\mathcal{S}(q,k)$ in Equation \eqref{eq:add_simplex_codes} and $\mathcal{S}(q,2k)$ in Equation \eqref{eq:reorder_added_simplex_codes}, respectively.
We denote by $\mathcal{C}_1$ the block code with generator matrix
$$\begin{pmatrix}
    G_0 & 0_{k \times \frac{q^k-1}{q-1}} & \cdots & 0_{k \times \frac{q^k-1}{q-1}}\\G_1 & S(q,k)
    & \cdots & S(q,k)
\end{pmatrix}$$
and  by $\mathcal{\hat{C}}_1$ the block code with generator matrix
$$\begin{pmatrix}
    \hat{G}_0 & 0_{k \times \frac{q^k-1}{q-1}} & \cdots & 0_{k \times \frac{q^k-1}{q-1}}\\\hat{G}_1 & S(q,k)
    & \cdots & S(q,k)
\end{pmatrix}.$$
Note that both codes have length $n_1=q^{\delta-k} \frac{q^{2k}-1}{q-1}$ and dimension $2k$.
Applying Lemma \ref{lemma:plotkin} to $\mathcal{\hat{C}}_1$ and  showing equality in the bound of Lemma \ref{lemma:plotkin} for $\mathcal{C}_1$ gives
\begin{equation}
\label{eq:plotkin_C_1}
    \sum_{c\in\mathcal{\hat{C}}_1}\operatorname{wt}(c)\leq |\mathcal{\hat{C}}_1| n_1 \left(1 - \frac{1}{q} \right)=q^{\delta-k} (q^{2k}-1) q^{2k-1}=
\sum_{c\in\mathcal{C}_1}\operatorname{wt}(c).
\end{equation}
Since $\mathcal{C}_1$ also is a one-weight code with all non-zero weights equal to $q^{\delta-k} q^{2k-1}$, it follows from \eqref{eq:plotkin_C_1} that
\begin{align}\label{G1}
    &q^{\delta+k-1}\geq \min_{c\in\mathcal{\hat{C}}_1\setminus\{0\}}\operatorname{wt}(c)=\nonumber\\
    &\min\left(q^{\delta+k-1},\min\left\{\operatorname{wt}\left((u_1,  u_0)\begin{pmatrix}
    \hat{G}_0 & 0_{k \times \frac{q^k-1}{q-1}} & \cdots & 0_{k \times \frac{q^k-1}{q-1}}\\\hat{G}_1 & S(q,k)
    & \cdots & S(q,k)
\end{pmatrix}\right), u_0\neq 0\right\}\right)
\end{align}
where the last equality is true since $\hat{G}_0$ is the generator matrix of a one-weight block code with minimum distance $q^{\delta+k-1}$.
Therefore,
\begin{align}\label{true}
    \min\left\{\operatorname{wt}\left((u_1,  u_0)\begin{pmatrix}
    \hat{G}_0 & 0_{k \times \frac{q^k-1}{q-1}} & \cdots & 0_{k \times \frac{q^k-1}{q-1}}\\
    \hat{G}_1 & S(q,k) &
    \cdots & S(q,k)
\end{pmatrix}\right), u_0\neq 0\right\}\leq q^{\delta+k-1}
\end{align} since if
\begin{align}\label{false}
\min\left\{\operatorname{wt}\left((u_1,  u_0)\begin{pmatrix}
    \hat{G}_0 & 0_{k \times \frac{q^k-1}{q-1}} & \cdots & 0_{k \times \frac{q^k-1}{q-1}}\\
    \hat{G}_1 & S(q,k) &
    \cdots & S(q,k)
\end{pmatrix}\right), u_0\neq 0\right\}> q^{\delta+k-1},
\end{align}
according to \eqref{G1}, one would have $\min_{c\in\mathcal{\hat{C}}_1\setminus\{0\}}\operatorname{wt}(c)=q^{\delta+k-1}$, which together with $\eqref{eq:plotkin_C_1}$ implies $\operatorname{wt}(c)=q^{\delta+k-1}$ for all $c\in\hat{\mathcal{C}}_1\setminus\{0\}$, contradicting \eqref{false}.
It follows that
\small
\begin{align*}
    &\min\left\{\operatorname{wt}\left((u_1, u_0)\begin{pmatrix}
    \hat{G}_0\\\hat{G}_1
\end{pmatrix}\right), u_0\neq 0\right\}\\
&=\hspace{-1.0mm}\min_{u_0\neq 0}\hspace{-0.6mm}\left\{\hspace{-0.6mm}\operatorname{wt}\hspace{-0.6mm}\left(\hspace{-1.1mm}(u_1, u_0)\hspace{-0.9mm}\begin{pmatrix}
    \hat{G}_0 & 0 & \cdots & 0\\
    \hat{G}_1 & S(q,k) &
     \cdots & S(q,k)
\end{pmatrix}\hspace{-1.2mm}\right)\hspace{-0.6mm}-\hspace{-0.6mm}\operatorname{wt}\hspace{-0.6mm}\left(\hspace{-0.8mm}(u_1, u_0)\hspace{-0.7mm}\begin{pmatrix}
    0 & \cdots & 0\\
    S(q,k) & \cdots & S(q,k)
\end{pmatrix}\hspace{-0.8mm}\right) \hspace{-0.6mm}\right\}\\
&=\min\left\{\operatorname{wt}\left((u_1, u_0)\begin{pmatrix}
    \hat{G}_0 & 0 & \cdots & 0\\\hat{G}_1 & S(q,k) &
    \cdots & S(q,k)
\end{pmatrix}\right), u_0\neq 0\right\}-q^{\delta-1}\leq q^{\delta+k-1}-q^{\delta-1}
\end{align*}
\normalsize
where the last inequality follows from \eqref{true}.
If $\hat{G}_0=G_0$ and $\hat{G}_1=G_1$ this bound is attained with equality.
This shows that the choice $G_0,G_1$ is optimal.
Moreover, for all $u_0\in\mathbb F_q^k\setminus\{0\}$, $u_1\in\mathbb F_q^k$, one has
$$\operatorname{wt}\left((u_1,  u_0)\begin{pmatrix}
    \hat{G}_0 & 0 & \cdots & 0\\\hat{G}_1 & S(q,k)
    & \cdots & S(q,k)
\end{pmatrix}\right)=q^{\delta+k-1}$$
and
$$\operatorname{wt}\left((u_1,  u_0)\begin{pmatrix}
    \hat{G}_0 \\\hat{G}_1
\end{pmatrix}\right)=q^{\delta+k-1}-q^{\delta-1}.$$
We can continue like this and assume that $\hat{G}_0,\hdots,\hat{G}_{j-1}$ for $j\leq\left\lfloor \frac{\delta}{k} \right\rfloor$ have been chosen to maximize $d^c_0,\hdots,d_{j-1}^c$ and
$$\operatorname{wt} (u_0\hat{G}_0)=q^{\delta+k-1}\quad \text{and}\quad
\operatorname{wt}\left((u_{t}, \ldots, u_0)
\left(\begin{array}{c}
    \hat{G}_0\\ \vdots \\ \hat{G}_{t}
\end{array}\right)\right)=q^{\delta+k-1}-q^{\delta-1}$$
for all $u_0\in\mathbb F_q^k\setminus\{0\}$, $u_1,\hdots,u_t\in\mathbb F_q^k$ for all $t \in \{1, \ldots, j-1 \}$.
Hence,
$$d_j^c=d_{j-1}^c+\min\left\{\operatorname{wt}\left((u_j, \ldots, u_0)\left(\begin{array}{c}
    \hat{G}_0\\ \vdots\\ \hat{G}_j
\end{array}\right)\right), u_0\neq 0\right\}.$$
Now we can reorder the columns of
$$\left(\begin{array}{cccc}
        G_0 & 0_{k \times \frac{q^{jk}-1}{q-1}} & \cdots & 0_{k \times \frac{q^{jk}-1}{q-1}}\\
        G_1\\ \vdots & S(q,jk) & \cdots & S(q,jk)\\ G_j
\end{array}\right),
$$
where we have $q^{\delta-jk}$ copies of $S(q,jk)$, to obtain the matrix
$$  \begin{pmatrix}
    S(q,(j+1)k) & \cdots & S(q,(j+1)k)
\end{pmatrix}$$
with $q^{\delta-jk}$ copies of $S(q,(j+1)k)$.
A block code with such a generator matrix has the property that all non-zero weights are equal to $q^{\delta-jk} q^{(j+1)k-1}$ and reaches the bound of Lemma \ref{lemma:plotkin} with equality.
Proceeding as before, we obtain
\small{\begin{align*}
    &\min_{u_0\neq 0}\left\{\operatorname{wt}\left((u_j, \ldots, u_0)\left(\begin{array}{c}
    \hat{G}_0\\  \vdots\\  \hat{G}_j
\end{array}\right)\right)\right\}\\
&=\hspace{-1.3mm}\min_{u_0\neq 0} \hspace{-0.9mm}\left(\hspace{-1.0mm} \operatorname{wt} \hspace{-0.9mm} \left( \hspace{-1.3mm}(u_j, \ldots, u_0)\hspace{-1.2mm}\left(\begin{array}{cc}
        \hat{G}_0 & 0_{k \times q^{\delta - jk} \frac{q^{jk}-1}{q-1}}\\
        \hat{G}_1\\ \vdots & \hspace{-1.5mm}S(q,jk)  \cdots  S(q,jk)\\ \hat{G}_j
   \end{array}\right)\hspace{-1.8mm}\right)\hspace{-1.2mm}-\hspace{-0.9mm}\operatorname{wt}\hspace{-0.9mm} \left(\hspace{-1.3mm}(u_j, \ldots, u_0)\hspace{-1.1mm}\begin{pmatrix}  0_{k \times q^{\delta - jk} \frac{q^{jk}-1}{q-1}}\\
          S(q,jk)  \cdots  S(q,jk)\\
   \end{pmatrix}\hspace{-1.5mm}\right)\hspace{-1.6mm}\right)\\
&=\min_{u_0\neq 0}\left\{\operatorname{wt}\left((u_j, \ldots, u_0)\left(\begin{array}{cc}
        \hat{G}_0 & 0_{k \times q^{\delta - jk} \frac{q^{jk}-1}{q-1}}\\
        \hat{G}_1\\ \vdots & S(q,jk)  \cdots  S(q,jk)\\ \hat{G}_j
   \end{array}\right)\right)\right\}-q^{\delta-1}\\
&\leq q^{\delta-jk} d(S(q,(j+1)k))-q^{\delta-1}=q^{\delta+k-1}-q^{\delta-1},
\end{align*}}
\normalsize
where to obtain the last inequality one can proceed analogously to \eqref{eq:plotkin_C_1}, \eqref{G1}, \eqref{true}.
Similarly as before, this bound is attained with equality if $\hat{G}_0=G_0, \ldots, \hat{G}_j=G_j$.
This shows that the choice $G_0, \ldots, G_j$ is optimal.
Moreover, for all $u_0\in\mathbb F_q^k\setminus\{0\}$, one has
$$\operatorname{wt}\left((u_j, \ldots, u_0)\left(\begin{array}{c}
    \hat{G}_0\\  \vdots\\  \hat{G}_j
\end{array}\right)\right)=q^{\delta+k-1}-q^{\delta-1}.$$
Since for $j>\left\lfloor\frac{\delta}{k}\right\rfloor$, $\hat{G}_j$ has at least one row equal to zero, the column distances cannot increase any further.
Since $\mathcal{C}$ is non-catastrophic, the formula for the free distance follows directly from $d_{\mathrm{free}}(\mathcal{C})=\lim_{j\rightarrow\infty}d_j^c(\mathcal{C})$.
\end{proof}

\begin{remark}
Since for Construction 1, for each $i\leq\mu-1$,
$$\operatorname{wt}\left((u_i, \ldots, u_0)\left(\begin{array}{c}
    G_0\\  \vdots\\  G_i
\end{array}\right)\right)$$
has the same value for all $(u_i,\hdots,u_0)\in\mathbb F_q^{k(i+1)}$ with $u_0\in\mathbb F_q^k\setminus\{0\}$, the order of the rows of any of the matrices $G_i$ for $i\in\{0,\hdots,\mu\}$ does not matter for the values of the column distances.
In particular, it does not matter which rows of $G_{\mu}$ we choose to be non-zero, i.e., which rows of $G_{\mu}$ we choose to form $\tilde{G}_{\mu}$.
Moreover, if we use Construction 1 for parameters $k$ and $\delta$ with $k\nmid\delta$, we obtain $d_j^c=d^c_{\mu-1}$ for $j\geq\mu-1$, i.e., none of the column distances depends on $\tilde{G}_{\mu}$. Therefore, considering instead the convolutional code from Construction 1 with the same value for $k$ but with $\bar{\delta}:=k \lfloor\frac{\delta}{k}\rfloor$, we obtain a convolutional code with the same column distances (and free distance) but larger rate. This implies that using Construction 1 for parameters with $k\nmid\delta$ is of no practical use.
\end{remark}

We illustrate Construction 1 with the following examples.

\begin{example}\label{ex:macDonald}
As in Example \ref{ex:simplex}, we choose $q=3$, $\delta=1$ and $k=2$, to obtain
    $$C_{3,1}(3)=\left(\begin{array}{cccccccccccc}
      1 & 1 & 1 & 0 & 1 & 1 & 1 & 0  & 1 & 1 & 1 & 0 \\
      0 & 1 & 2 & 1  & 0 & 1 & 2 & 1  & 0 & 1 & 2 & 1  \\
      \hline
    0 & 0 & 0 & 0  & 1 & 1 & 1 & 1  & 2 & 2 & 2 & 2
  \end{array}\right)=\begin{pmatrix}
      G_0\\ \tilde{G}_1
  \end{pmatrix},$$
  i.e., $G(z)=G_0+G_1z$ with
\begin{align*}
G_0&=\left(\begin{array}{cccccccccccc}
      1 & 1 & 1 & 0 & 1 & 1 & 1 & 0  & 1 & 1 & 1 & 0 \\
      0 & 1 & 2 & 1  & 0 & 1 & 2 & 1  & 0 & 1 & 2 & 1
  \end{array}\right)\\
   G_1&=\left(\begin{array}{cccccccccccc}
          0 & 0 & 0 & 0  & 1 & 1 & 1 & 1  & 2 & 2 & 2 & 2 \\
          0 & 0 & 0 & 0 & 0 & 0 & 0 &0 & 0 & 0 & 0 & 0
  \end{array}\right)
\end{align*}
is the generator matrix of a $(12,2,1)$ convolutional code over $\mathbb F_3$ with optimal column distances $d^c_j=9$ for all $j\in\mathbb{Z}_{\geq 0}$.
Note that a $(12,2,1)$ convolutional code is MDP if and only if $d^c_0=11$ (since in Definition \ref{def:mdp} we have $L=0$ for these parameters), i.e., $G_0$ has to be the generator matrix of an MDS block code and consequently $q\geq 11$ since the MDS conjecture holds in this case as is easy to observe, see for example \cite{achninemain}, \cite{bush1952orthogonal}.
\end{example}

\begin{example}\label{ex:k1}
  Choosing $q=3$, $k=1$, $\delta=2$, we obtain
 $$C_{3,2}(3)=R(3,2)=\left(\begin{array}{ccccccccc}
      1 & 1 & 1 & 1 & 1 & 1 & 1 & 1 & 1  \\
      \hline
      0 & 1 & 2   & 0 & 1 & 2 &  0 & 1 & 2   \\
     \hline
    0 & 0 &  0   & 1 & 1 & 1  & 2 & 2 & 2
  \end{array}\right)=\begin{pmatrix}
      G_0\\ G_1\\ G_2
  \end{pmatrix},$$
  i.e., $G(z)=G_0+G_1z+G_2z^2$ with
\begin{align*}
G_0&=\left(\begin{array}{ccccccccc}
       1 & 1 & 1 & 1 & 1 & 1 & 1 & 1 & 1
  \end{array}\right)\\
   G_1&=\left(\begin{array}{ccccccccc}
          0 & 1 & 2   & 0 & 1 & 2 &  0 & 1 & 2
  \end{array}\right)\\
  G_2&=\left(\begin{array}{ccccccccc}
          0 & 0 &  0   & 1 & 1 & 1  & 2 & 2 & 2
  \end{array}\right)
\end{align*}
is the generator matrix of a $(9,1,2)$ convolutional code over $\mathbb F_3$ with optimal column distances $d^c_0=9$, $d^c_1=15$, and $d^c_j=21$ for all $j\geq 2$. In comparison with this a $(9,1,2)$ convolutional code is MDP if and only if $d^c_0=9$, $d^c_1=17$ and $d^c_2=25$ (since in Definition \ref{def:mdp} we have $L=2$ for these parameters).
To the best of our knowledge, among the existing constructions for $(9,1,2)$ MDP convolutional codes, the one from \cite{rscc} requires the smallest possible field size equal to $121$.
Hence, the field size for Construction 1 is much smaller than the one for the best known MDP code construction with the same parameters $(n, k, \delta)$, while almost reaching the same column distances.
\end{example}

\begin{remark}
In general it is an open problem to construct MDP convolutional codes over small finite fields, especially for values $L\geq 2$ in Definition \ref{def:mdp}.
There are several papers reducing the field size for the special case $L=1$ (see e.g. \cite{cheng2026new, chen}), which require $q\approx n^{\delta}$. In comparison to this, Construction 1 has $q\approx n^{1/(\delta+k-1)}$, i.e., compared to constructions for MDP convolutional codes with the same code parameters we achieve a considerable gain with respect to the field size.
\end{remark}

\begin{remark}
    It is easy to see that multiplications of columns of $(G_0^\top \cdots\ G_{\mu-1}^\top \ \tilde{G}_{\mu}^\top)^\top$ in Theorem \ref{thm:construction_1} with non-zero elements from $\mathbb F_q$ and column permutations do not affect the column distances.
    Therefore, we can also multiply the generator matrix $C_{\delta+k, \delta}(q)$ by a monomial matrix on the right and still get a convolutional code with optimal column distances.

    Actually, for $k\mid\delta$, the proof of Theorem \ref{thm:construction_1} gives also the converse of this observation.
    Namely any convolutional code with optimal column distances with the parameters of Construction $1$ comes from a generator matrix $C_{\delta+k, \delta}(q)$ multiplied with a monomial matrix on the right.
    We state this as the next theorem.
    This also completes the proof for binary codes in \cite{isit23} where this step was missing.
\end{remark}

\begin{theorem}
    For $k\mid\delta$, any $\left(\frac{q^{\delta}(q^k-1)}{q-1},k,\delta\right)$ convolutional code over $\mathbb{F}_q$ with optimal column distances comes from Construction $1$ using generator matrices $C_{\delta+k, \delta}(q)$ of MacDonald codes multiplied with a monomial matrix on the right.
\end{theorem}
\begin{proof}
The proof of Theorem \ref{thm:construction_1} shows that for any convolutional code with optimal column distances and  generator matrix $\hat{G}(z)=\sum_{i=0}^{\mu}\hat{G}_iz^i$, the matrix
$$\hat{S}:=\left(\begin{array}{cc}
        \hat{G}_0 & 0_{k \times \frac{q^\delta-1}{q-1}}\\
        \hat{G}_1\\ \vdots & S(q,\delta)\\ \hat{G}_{\mu}
   \end{array}\right)
  $$
is the generator matrix of a one-weight code of dimension $\delta + k$ and length $\frac{q^{\delta + k} - 1}{q-1}$.
Hence, by Bonisoli's Theorem \cite{bonisoli1984every} it is a simplex code of dimension $\delta + k$ and therefore equivalent to the code with generator matrix
$$S:=\left(\begin{array}{cc}
        G_0 & 0_{k \times \frac{q^\delta-1}{q-1}}\\
        G_1\\ \vdots & S(q,\delta)\\ G_{\mu}
   \end{array}\right).
  $$
Hence, by Lemma \ref{lemma:simplex_code_monomial_right} we have
$$\hat{S}M=S$$
with $M\in\mathbb F_q^{\frac{q^{\delta+k}-1}{q-1}\times \frac{q^{\delta+k}-1}{q-1}}$ a monomial matrix.
Due to the structure of the matrices $\hat{S}$ and $S$ we see that $M$ is of the form
\[
    M =
    \begin{pmatrix}
        M_n & 0_{n \times \frac{q^\delta-1}{q-1}}\\
        0_{\frac{q^\delta-1}{q-1} \times n} & I_{\frac{q^\delta -1}{q-1}}
    \end{pmatrix}
\]
with $M_n \in\mathbb F_q^{n \times n}$ a monomial matrix where $n = \frac{q^{\delta}(q^k-1)}{q-1}$.
This implies
\[
    \left(\begin{array}{c}
        \hat{G}_0\\
        \hat{G}_1\\ \vdots\\ \hat{G}_{\mu}
   \end{array}\right) M_n
   = \left(\begin{array}{c}
        G_0\\
        G_1\\ \vdots\\ G_{\mu}
   \end{array}\right),
\]
as we wanted to show.
\end{proof}

\begin{remark}
    Similarly as in \cite{isit23} taking concatenations  \[
        \begin{pmatrix}
        C^{(1)}_{\delta+k,\delta}(q) & \cdots &  C^{(\xi)}_{\delta+k,\delta}(q)
    \end{pmatrix}
    \]
     of generator matrices $C^{(i)}_{\delta+k,\delta}(q)$ of MacDonald codes
     with the form as in Construction $1$ for $i=1,\hdots, \xi$ gives convolutional codes with optimal column distances.
    The decoding can also easily be adapted to this case.
\end{remark}

\begin{remark}
    Note that not every generator matrix of a MacDonald code gives optimal column distances.
    The reason is that row operations on the generator matrix of a MacDonald code can change the column distances of the corresponding convolutional code.
    In particular, using Lemma \ref{lemma:weight_Mac_Donald_Codes} one can directly construct a generator matrix of a MacDonald code such that the $0$-th column distance of Construction $1$ is not optimal.
\end{remark}

The fact that Construction 1 has MacDonald codes and therefore, according to \eqref{eq:MacDonald_RM}, also first order Reed-Muller codes as building blocks, will allow us to present an improved version of the Viterbi decoding algorithm for these codes, making use of an efficient decoding algorithm for first order Reed-Muller codes.
Note that a similar approach for improving the complexity of the Viterbi algorithm for convolutional codes that are based on efficiently decodable block codes, has been used in \cite{isit24} for the special case $q=2$.
In that paper Hadamard matrices have been used for the decoding of binary first order Reed-Muller codes.
However, for general $q$, we employ a different decoding algorithm for first order Reed-Muller and MacDonald codes, which we will present next.

\section{Efficient decoding of MacDonald codes}
\label{sec:decoding}

In this section, we first recall some results on the decoding of first order Reed-Muller codes from \cite{AshikhminL96} and then apply them to the decoding of MacDonald codes and certain subcodes of MacDonald codes.
We will use the results of this section to present a reduced complexity Viterbi algorithm for the codes of Construction 1 in the next section.
Note that in order to be able to use these techniques later together with the Viterbi algorithm, we need to have decoding algorithms that calculate the distances of all codewords to a received word.

\subsection{Efficient decoding of first order Reed-Muller Codes}

This subsection follows \cite{AshikhminL96}, but we provide additional details.

We assume that we have some fixed bijection of the elements in $\mathbb F_q$ with the elements in $\{0, 1, \ldots, q-1\}$.
Let us call the corresponding map $\phi$.
This will be used in conjunction with some $q$-ary expansions.

We fix a generator matrix $G$ of a first order Reed-Muller code $\mathcal{R}(q, m)$ as follows.
We evaluate the functions $x_1, \ldots, x_m, 1$ at all the points in $\mathbb F_q^m$ (in some particular order of the points which will be explained later).
We call the rows of this generator matrix $g_0, g_1, \ldots, g_m$, i.e.,
\[
    G =
    \left(\begin{array}{c}
        g_0\\
        g_1\\
        \vdots\\
        g_m
    \end{array}\right).
\]
The matrix $G$ has the all one vector as last row $g_m$ and the matrix formed from the rows $g_0, g_1, \ldots, g_{m-1}$ has as columns all the vectors in $\mathbb F_q^m$.
We will call this matrix $G'$, i.e.,
\[
    G' =
    \left(\begin{array}{c}
        g_0\\
        g_1\\
        \vdots\\
        g_{m-1}
    \end{array}\right).
\]
Using the identification $\phi$ of $\mathbb F_q$ with the set $\{0, 1, \ldots, q-1\}$ from above we can order the columns in $G'$ so that we have a lexicographic order of the columns.
We use the same order of the columns for $G$.
To be precise, we fix the generator matrix $G$, which has the following form under the identification $\phi$:
\[
    \left(\begin{array}{cccccccc}
        0 & 1 & \ldots & q-1 & 0 & 1 & 2 &\ldots \\
        0 & 0 & \ldots & 0 & 1 & 1 & 1 &\ldots\\
        \vdots & \vdots & \vdots & \vdots & \vdots & \vdots & \vdots & \vdots \\
        1 & 1 & 1 & 1 & 1 & 1 & 1 & \ldots
     \end{array}\right).
\]
We define the code $\mathcal{R}'(q, m)$ to be the code with generator matrix $G'$.
We will mainly work with this code.
Each codeword $c \in \mathcal{R}'(q, m)$ can be identified with a number in $\{0,\dots,  q^m - 1\}$ as follows.
Write $c = \sum_{i = 0}^{m-1} \lambda_i g_i$ with $\lambda_i \in \mathbb F_q$ and identify $\lambda_i$ with the corresponding number $\phi(\lambda_i)$ in $\{0,\dots, q-1\}$.
Then we can set $\psi(c) = \sum_{i = 0}^{m-1} \phi(\lambda_i) q^i$.
This is obviously a bijection between codewords in $\mathcal{R}'(q, m)$ and numbers in $\{0,\hdots, q^m - 1\}$ via $q$-ary expansions.
\begin{definition}
\label{def:Bz}
    Let $B$ be the $q^m \times q^m$ matrix of all codewords of $\mathcal{R}'(q, m)$ where we order the rows $b_i$ according to the size of $\psi(b_i)$ in increasing order.
    So for $$i = i_0 + i_1 q + \ldots +i_{m-1} q^{m-1}$$ we can write the $i$-th row $b_i$ of $B$ as
\begin{equation}
\label{eq:lexicographic_order_rows}
    b_i = \phi^{-1}(i_0) g_0 + \ldots + \phi^{-1}(i_{m-1}) g_{m-1}.
\end{equation}
We define a matrix $B_z$ of size $q^m \times q^m$ where $z$ can formally have exponents from $\mathbb F_q$ as follows
\[
    (B_z)_{ij} = z^{b_{ij}}.
\]
\end{definition}
So technically we are dealing with the group ring $\mathbb Z[\mathbb F_q]$.
Next, we define an action of $\mathbb Z[\mathbb F_q]$ on $\mathbb R^q$ as follows.
Let $s = (s_0, \ldots, s_{q-1}) \in \mathbb{R}^q$.
For $c \in \mathbb{F}_q$ we set
\begin{equation}
\label{eq:def_scalar_star}
    z^c \star s = (s_{\phi(\phi^{-1}(0) + c)}, s_{\phi(\phi^{-1}(1) + c)}, \ldots, s_{\phi(\phi^{-1}(q-1) + c)}).
\end{equation}
One can check that $(z^a z^b) \star s = z^a \star (z^b \star s)$.
Furthermore, we also get $z^a \star (s_1 + s_2) = z^a \star s_1 + z^a \star s_2$ and $z^0 \star s = s$.
Finally, $z^a \star (\mu s) = \mu (z^a \star s)$ for $\mu \in \mathbb{R}$.
We also define the $\star$-multiplication of two vectors as follows
\begin{align*}
    (z^{c_0}, \ldots, z^{c_{n-1}}) \star (s^{(0)}, \ldots, s^{(n-1)})
    = \sum_{i = 0}^{n-1} z^{c_i} \star s^{(i)},
\end{align*}
where $c_i \in \mathbb{F}_q$ and $s^{(i)} \in \mathbb{R}^q$.
Similarly we define the matrix multiplication.
Let $E = (e_{ij})$
be a matrix whose entries are of the form $\lambda z^{c_{ij}}$ with $\lambda \in \mathbb{Z}, c_{ij} \in \mathbb{F}_q$ and $D = (d_{ij})$ be a matrix with $d_{ij} \in \mathbb{R}^q$.
Then we define
\begin{equation}
\label{eq:star_prod_mat}
    (E\star D)_{ij} = \sum_{k} e_{ik} \star d_{kj}.
\end{equation}
The following lemma can be shown by calculations.
\begin{lemma}
\label{lemma:star_action}
    The matrices $E$ with entries in $\mathbb{Z}[\mathbb{F}_q]$ from above act as a semigroup on matrices of the form of $D$ as above with entries $d_{ij} \in \mathbb{R}^q$.
\end{lemma}
We will use the definition from equation \eqref{eq:star_prod_mat} in the following to find the distances of all codewords of a first order Reed-Muller code to the received word.

Let $v \in \mathcal{R}(q, m)$ be the transmitted word and $w \in \mathbb{F}_q^{n}$ with $n = q^m$ be the received word.
We write $w_i = v_i + e_i$ with $e_i \in \mathbb{F}_q$ the error in position $i$.
To $w$ we associate the vector $S_w = (s_w^{(0)}, \ldots, s_w^{(n-1)})$ where $s_w^{(i)} = (s_{w, 0}^{(i)}, \ldots, s_{w, q-1}^{(i)})$ with
\begin{align}
\label{sw}
    s_{w, j}^{(i)} = \begin{cases}
         1, & \text{ if $w_i = \phi^{-1}(j)$},\\
         0, & \text{ else.}
    \end{cases}
\end{align}
In addition, for $v$ we define a vector $v_z = (z^{v_0}, \ldots, z^{v_{n-1}})$.
Finally, we let
$$N(v, w) = |\{i \in \{0, \hdots, n-1\} : v_i = w_i\}|.$$
Then we get the following lemma.

\begin{lemma}\cite{AshikhminL96}
\label{lemma:compute_correlatin_vz_Sw}
    With the notation from above let $v_z \star S_w = (k_0, \ldots, k_{q-1})$.
    Then $k_\ell = N(v + E_\ell, w)$ where $E_\ell = (\phi^{-1}(\ell), \ldots, \phi^{-1}(\ell))$.
\end{lemma}
\begin{proof}
    By definition we have
    \begin{align*}
        k_\ell &= (v_z \star S_w)_\ell = \left(\sum_{j=0}^{n-1} z^{v_j} \star s_w^{(j)} \right)_\ell\\
        &= \left(\sum_{j=0}^{n-1} \left(s_{w, \phi(\phi^{-1}(0) + v_j)}^{(j)}, \ldots, s_{w,\phi(\phi^{-1}(q-1) + v_j)}^{(j)}\right) \right)_\ell= \sum_{j=0}^{n-1} s_{w, \phi(\phi^{-1}(\ell) + v_j)}^{(j)}.
    \end{align*}
    Now, let $I = \{a \in \{0,\dots,n-1\} : v_a + \phi^{-1} (\ell) = w_a \}$.
    Then
    \[
        k_\ell = \sum_{j \in I} s_{w, \phi(\phi^{-1}(\ell) + v_j)}^{(j)} + \sum_{j \notin I} s_{w, \phi(\phi^{-1}(\ell) + v_j)}^{(j)}.
    \]
    Note that if $j \in I$, we have $\phi^{-1}(\ell) + v_j = w_j$ and hence $s_{w, \phi(\phi^{-1}(\ell) + v_j)}^{(j)} = s_{w, \phi(w_j)}^{(j)} = 1$ by definition.
    On the other hand, if $\phi^{-1}(\ell) + v_j \neq w_j$, then $s_{w, \phi(\phi^{-1}(\ell) + v_j)}^{(j)} = 0$ by the definition of $S_w$ and since $\phi$ is a bijection.
    Therefore,
    \[
        k_\ell = \sum_{j \in I} s_{w, \phi(\phi^{-1}(\ell) + v_j)}^{(j)} = |I| = N(v + E_\ell, w).
    \]
    This concludes the proof.
\end{proof}
We obtain the following corollary for $B_z$ from Definition \ref{def:Bz}.

\begin{corollary}\cite{AshikhminL96}
   The distances between a received word $w$ and all codewords of $\mathcal{R}(q,m)$ can be
    found by computing $B_z \star S_w$.
\end{corollary}

We have seen that $B_z$ plays a crucial role in decoding.
Now, we will decompose $B_z$ as in \cite{AshikhminL96} into a product of sparse matrices in order to get a fast decoding algorithm.
The following results are well known, but we provide them in order for the paper to be self-contained.
\begin{definition}
    Let $A$ be an $m_1 \times n_1$ matrix and $B$ be another matrix.
    For $i \in \{0, \ldots, m_1-1\}$ and $j \in \{0, \ldots, n_1-1\}$ we denote by $A_{i,j}$ the entry of $A$ in position $i, j$.
    We define the Kronecker product of the matrices $A, B$ denoted by $A \otimes B$ as follows:
\[
    A \otimes B = \left(\begin{array}{ccc}
        A_{0,0} B & \ldots & A_{0,n_1-1} B \\[0.2cm]
        \vdots & \ddots & \vdots\\[0.2cm]
        A_{m_1-1,0} B & \ldots & A_{m_1-1,n_1-1} B
    \end{array}\right).
\]
\end{definition}
From this definition we get the following lemma.

\begin{lemma}
    Let $A$ be an $m_1 \times n_1$ matrix and $B$ be an $m_2 \times n_2$ matrix.
    Write $i = i_1 m_2 + i_2$ with $0 \leq i_2 < m_2$ and $j = j_1 n_2 + j_2$ with $0 \leq j_2 < n_2$.
    Then $(A \otimes B)_{ij} = A_{i_1,j_1} B_{i_2,j_2}$.
\end{lemma}

We can specialize the previous lemma to the following setting:
\begin{lemma}
\label{lemma:tensor_power}
    Let $H$ be a $q \times q$ matrix.
    Let $i = i_0 + i_1 q + \ldots + i_{m-1} q^{m-1}$ and $j = j_0 + j_1 q + \ldots + j_{m-1} q^{m-1}$ be $q$-ary expansions.
    Then
    \[
        (H^{\otimes m})_{i,j} = \prod_{\ell=0}^{m-1} H_{i_\ell,j_\ell}.
    \]
\end{lemma}
\begin{proof}
    We prove the claim by induction.
    It is obvious for $m = 1$.
    Assume the claim is proven for $H^{\otimes m-1}$.
    Then, $H^{\otimes m} = H^{\otimes m-1} \otimes H$.
    Applying the previous lemma we write $i = i'_1 q + i'_2$ with $0 \leq i'_2 < q$ and $j = j'_1 q + j'_2$ with $0 \leq j'_2 < q$ to get
    \[
        (H^{\otimes m})_{i,j} = (H^{\otimes m-1})_{i'_1,j'_1} H_{i'_2,j'_2}.
    \]
    Note that $i'_2 = i_0, j'_2 = j_0$.
    Furthermore, $i'_1 = i_1 + i_2 q + \ldots + i_{m-1} q^{m-2}$ and $j'_1 = j_1 + j_2 q + \ldots + j_{m-1} q^{m-2}$.
    Applying the induction hypothesis yields the result.
\end{proof}

This allows us to write the matrix $B_z$ as $H^{\otimes m}$ for a suitable $H$ as we show in the next lemma following \cite{AshikhminL96}.
Recall that $\phi$ was the bijection we fixed between $\mathbb{F}_q$ and $\{0, \ldots, q-1\}$.
\begin{lemma}
\label{lemma:tensor_power_H_Bz}
    Let $n = q^m$.
    We define the matrix $H$ of size $q \times q$ by
    \[
        H_{ij} = z^{\phi^{-1}(i) \phi^{-1}(j)}
    \]
    for $i, j \in \{0, \ldots, q-1\}$.
    Then,
    \[
        B_z = H^{\otimes m}.
    \]
\end{lemma}
\begin{proof}
    To rewrite $H^{\otimes m}$ we use Lemma \ref{lemma:tensor_power}.
    Writing $i = i_0 + i_1 q + \ldots + i_{m-1} q^{m-1}$ and $j = j_0 + j_1 q + \ldots + j_{m-1} q^{m-1}$ we get
    \[
        (H^{\otimes m})_{i, j} = z^{\phi^{-1}(i_0) \phi^{-1}(j_0) + \ldots +\phi^{-1}(i_{m-1}) \phi^{-1}(j_{m-1})}.
    \]
    Now we have to determine how we can express $B_z$.
    The $i$-th row $b_i$ of $B$ is $b_i = \phi^{-1}(i_0) g_0 + \ldots + \phi^{-1}(i_{m-1}) g_{m-1}$.
    Since the columns of $G'$ are ordered lexicographically, we get that $g_{\ell, j} = \phi^{-1}(j_\ell)$.
    Thus,
    \[
        B_{i,j} = \phi^{-1}(i_0) \phi^{-1}(j_0) + \ldots + \phi^{-1}(i_{m-1}) \phi^{-1}(j_{m-1}),
    \]
    and the claim follows.
\end{proof}
Note also the similarity to the discrete Fourier transform as mentioned in \cite{AshikhminL96}.
From now on when we write $H$, we use the matrix $H$ from Lemma \ref{lemma:tensor_power_H_Bz}.
We can write $H^{\otimes m}$ as
\begin{align*}
    H^{\otimes m} &= H \otimes H \otimes \cdots \otimes H \\
    &= (H \otimes I \otimes I \otimes \cdots \otimes I) (I \otimes H \otimes I \otimes \cdots \otimes I) \cdots (I \otimes I \otimes \cdots \otimes I \otimes H).
\end{align*}

We want to rewrite this equation.
In order to do so, we will use the results from Davio \cite{davio2012kronecker}, see also \cite{diaconis1983mathematics}, and introduce generalized perfect shuffles.

\begin{definition}[generalized perfect out shuffle]
    Let $\sigma, \tau \in \mathbb{Z}_{\geq 1}$.
    The \textbf{out-$\sigma$-shuffle} depending on parameter $\tau$ is a permutation
    $\pi : \{0, \ldots, \sigma \tau - 1\} \to \{0, \ldots, \sigma \tau -1\}$ defined as follows.
    Let $i = i_1 \sigma + i_2 \in \{0, \ldots, \sigma \tau -1\}$ with $0 \leq i_2 < \sigma$.
    Then we define
    \[
        \pi(i) = i_2 \tau + i_1.
    \]
\end{definition}
For $\sigma = 2$ this is the classical Faro out-shuffle.
It can be interpreted as follows.
Suppose we have $\sigma \tau$ cards.
The out-$\sigma$-shuffle is obtained by taking $\sigma$ stacks of $\tau$ cards and putting them next to each other.
Then we pick up the cards from left to right.

\begin{example}
    Let $\sigma = 3, \tau = 4$.
    Then we number the cards with the numbers $\{0, \ldots, 11\}$.
    The stacks are $\{0, 1, 2, 3\}, \{4, 5, 6, 7\}, \{8, 9, 10, 11\}$.
    Then the resulting permutation is given by
    \[
        \begin{bmatrix}
            0 & 1 & 2 & 3 & 4 & 5 & 6 & 7 & 8 & 9 & 10 & 11\\
            0 & 4 & 8 & 1 & 5 & 9 & 2 & 6 & 10 & 3 & 7 & 11
        \end{bmatrix}.
    \]
\end{example}
Davio \cite{davio2012kronecker}, see also \cite{diaconis1983mathematics}, proved a more general result, but we get in the special case of square matrices $A, B$ of sizes $\sigma \times \sigma$ and $\tau \times \tau$ respectively, that
\[
    B \otimes A = P_{\tau, \sigma} (A \otimes B) P_{\sigma, \tau}
\]
where $P_{\sigma, \tau}$ is the permutation matrix corresponding to the out-$\sigma$-shuffle on $\sigma \tau$ elements.
It is not too hard to see that $P_{\tau,\sigma} = P_{\sigma, \tau}^\top$, see for example \cite{medvedoff1987groups}, \cite[Lemma 3]{johnson2022look}.
From now on we will only write $P$ for $P_{\sigma,\tau}$ because for the complexity of decoding it only matters that $P$ is a permutation.

Specializing this, we have the permutation matrix $P$ such that for all $q^{m-1} \times q^{m-1}$ matrices $A$ and for all $q \times q$ matrices $B$, we have $B \otimes A = P (A \otimes B) P^\top$.
Thus,
\[
    I \otimes H \otimes I \otimes \cdots \otimes I = P (H \otimes I \otimes I \otimes \cdots \otimes I) P^\top,
\]
where we used for $B$ the rightmost $I$ in $H \otimes I \otimes I \otimes \cdots \otimes I \otimes I = (H \otimes I \otimes I \otimes \cdots \otimes I) \otimes I$.
Iteratively, using the rightmost matrix in the Kronecker product, we get
\begin{align*}
    I \otimes I \otimes H \otimes I \otimes \cdots \otimes I &= P^2 (H \otimes I \otimes I \otimes \cdots \otimes I) (P^2)^\top\\
    &\vdots\\
    I \otimes I \otimes I \cdots \otimes I \otimes H &= P^{m-1} (H \otimes I \otimes I \otimes \cdots \otimes I) (P^{m-1})^\top.
\end{align*}
In particular, returning to $H^{\otimes m}$ we get
\begin{equation}
\label{eq:factoriazion_kronecker_perm}
    H^{\otimes m} = ((H \otimes I \otimes \cdots \otimes I) P)^m.
\end{equation}
Since $H$ is a $q \times q$ matrix and $m = \log_q(n)$, we get the following theorem.

\begin{theorem}\label{decoding}\cite{AshikhminL96}
\label{cor:complexity_decoding_RM}
    Let $n = q^m$.
    The distances between a received word $w$ and all codewords of $\mathcal{R}(q,m)$ can be
    found with at most $q(q-1) n \log_q(n)$ additions.
\end{theorem}
\begin{proof}
    Since $\star$ is an action by Lemma \ref{lemma:star_action}, we get from Equation \eqref{eq:factoriazion_kronecker_perm} that
    \begin{align*}
        B_z \star S_w &= (H^{\otimes m}) \star S_w
        = ((H \otimes I \otimes \cdots \otimes I) P)^m \star S_w\\
        &= ((H \otimes I \otimes \cdots \otimes I) P) \star (((H \otimes I \otimes \cdots \otimes I) P)\star(\cdots \star S_w))
    \end{align*}
    where we have $m$ times the matrix $((H \otimes I \otimes \cdots \otimes I) P)$ acting on vectors in $(\mathbb{R}^q)^n$.
    The first action is on $S_w$ as defined in \eqref{sw}.
    Therefore, for the total complexity we get $m = \log_q(n)$ times the complexity of a single action with $\star$ by $(H \otimes I \otimes \cdots \otimes I) P$.
    Note that the matrix $(H \otimes I \otimes \cdots \otimes I)$ consists of blocks of identity matrices of size $q^{m-1}$ multiplied by entries of $H$.
    Therefore, it has only $q^{m-1} \times q^2 = n q$ non-zero entries and in each row and column there are exactly $q$ non-zero entries.
    Recall that $S_w = (s_w^{(0)}, \ldots, s_w^{(n-1)})$ with $s_w^{(i)} \in\ \mathbb R^q$ and by definition of $\star$ for matrix times vector multiplication we get in total $n$ times the complexity of computing a single entry of the form
    \begin{align*}
        (((H \otimes I \otimes \cdots \otimes I) P) \star S_w)_{i}
        = \sum_{k=0}^{n-1} ((H \otimes I \otimes \cdots \otimes I) P)_{ik} \star s_w^{(k)}.
    \end{align*}
    for $i \in \{0, \ldots, n-1\}$.
    But a row of the matrix $(H \otimes I \otimes \cdots \otimes I) P$ has only $q$ non-zero entries and hence we only have $q - 1$ instead of $n-1$ additions of star products by scalars as in \eqref{eq:def_scalar_star}.
    In order to add these $q-1$ vectors of length $q$, we need $(q-1)q$ additions.
    Hence, we get that overall we need at most $q (q-1) n \log_q(n)$ additions.
\end{proof}

\subsection{Comparing a received word to all codewords of a MacDonald code}

In this subsection, we use the relation of first order Reed-Muller codes and MacDonald codes together with the decoding algorithm from \cite{AshikhminL96} (see previous subsection) to present an efficient decoding algorithm for MacDonald codes. We assume these results to be known but could not find a reference where this decoding algorithm for MacDonald codes is explained.

According to Theorem \ref{decoding} (see \cite{AshikhminL96}), to compare all codewords of $\mathcal{R}(q,m)$ to a received word $w$, we need to calculate $B^{(m)}_z\star S_w$, where $B^{(m)}$ is the matrix of all codewords of $\mathcal{R}'(q,m)$ in a certain order.
To construct the matrices $B^{(m)}$ for certain values of $m$, we first construct a generator matrix $G'$ of $\mathcal{R}'(q,m)$ as above and then form linear combinations in lexicographical order, see  \eqref{eq:lexicographic_order_rows}.
For simplicity here we omit the bijection $\phi$ from the previous subsection.

We now consider a generator matrix of the MacDonald code $\mathcal{C}_{m,m-k}(q)$ in the following form (see \eqref{eq:MacDonald_RM}):
\begin{align}\label{eq:definition_R}
C_{m,m-k}(q)=\left(
    \begin{array}{c|c|c|c|c}
        R(q,m-1)   & 0_{1 \times q^{m-2}}            & 0_{2 \times q^{m-3}} & \ldots & 0_{(k-1)\times q^{m-k}}\\[0.2cm]
                                & R(q,m-2)  &  & \ldots &   \\[0.2cm]
                                &                       & R(q,m-3) &\ldots &   \\[0.2cm]
                          &             &     &\vdots     & R(q,m-k)
    \end{array}
    \right)
\end{align}
where the first row of $R(q,m-i)$ is the all-one row for each $i\in\{1,\hdots,k\}$.
Moreover, we assume that the columns of $R'(q,m-i)$, i.e., the matrix $R(q,m-i)$ without the first row, are in lexicographical order.

For a codeword $v\in\mathcal{C}_{m,\delta}(q)$ with $m=\delta+k$, we have $v=(v_1, \ldots, v_k)$ with $v_i\in\mathbb F_q^{q^{m-i}}$, where $v_i$ is a codeword of $\mathcal{R}(q,m-i)$.
We partition a received word $w$ in the same way, i.e., we write it as $w=(w_1, \ldots, w_k)$ with $w_i\in\mathbb F_q^{q^{m-i}}$.
We obtain $$N(v,w)=\sum_{i=1}^kN(v_i,w_i)$$ and to obtain $N(v_i,w_i)$ for all $v$, we need to calculate $B^{(m-i)}_z\star S_{w_i}$.

More precisely, if we write $\alpha^{(j)}=(\alpha^{(j)}_1, \ldots, \alpha^{(j)}_k)$ for the $j$-th row of $C_{m,m-k}(q)$, then
\small
\begin{align*}
&B^{(m-i)}_z\star S_{w_i}\\[0.2cm]
&=\setlength\arraycolsep{3.5pt}
\left(\begin{array}{cccc}
    N(0,w_i) & N(\alpha^{(i)}_i,w_i) & \cdots & N((q-1)\alpha^{(i)}_i,w_i)\\[0.2cm]
    N(\alpha^{(i+1)}_{i},w_i) &  N(\alpha^{(i+1)}_{i}+\alpha^{(i)}_i,w_i) & \cdots & N(\alpha^{(i+1)}_{i} \hspace{-1mm}+ \hspace{-0.5mm}(q-1)\alpha^{(i)}_i,w_i)\\[0.2cm]
    \vdots & \vdots & \cdots & \vdots\\[0.2cm]
   \hspace{-0.5mm}N((q-1) \hspace{-1.5mm}\displaystyle\sum_{s=i+1}^{m} \hspace{-1.5mm} \alpha_i^{(s)},w_i) &  N((q-1) \hspace{-1.5mm}\displaystyle\sum_{s=i+1}^{m} \hspace{-1.5mm} \alpha_i^{(s)} \hspace{-0.5mm}+ \hspace{-0.5mm}\alpha^{(i)}_i,w_i) & \cdots &  N((q-1)\displaystyle\sum_{s=i}^{m} \alpha_i^{(s)},w_i)
 \end{array}\right)
\end{align*}
\normalsize

Note that $\alpha^{(1)}_i=\ldots=\alpha^{(i-1)}_i=0$ and $\alpha_i^{(i)}$ is equal to the all one row $E_\ell$ of suitable length from Lemma \ref{lemma:compute_correlatin_vz_Sw} for some $\ell$.
Hence, for all codewords $v\in\mathcal{C}_{m,\delta}(q)$, the value for $N(v_i,w_i)$ can be found as an entry of the matrix $B^{(m-i)}_z\star S_{w_i}$.

\begin{example}
    Take $q=3$, $k=2$, $\delta=1$. One has $$C_{3,1}(3)=\begin{pmatrix}
        R(3,2) & 0_{1 \times 3} \\[0.2cm]
         &  R(3,1)
    \end{pmatrix}=\left(\begin{array}{cccccccccccc}
       1 & 1 & 1 & 1 & 1 & 1 & 1 & 1 & 1 & 0 & 0 & 0\\[0.2cm]
        0 & 1 & 2 & 0 & 1 & 2 & 0 & 1 & 2 & 1 & 1 & 1\\[0.2cm]
        0 & 0 & 0 & 1 & 1 & 1 & 2 & 2 & 2 & 0 & 1 & 2
    \end{array}\right)$$
and
$$B^{(2)}=\left(\begin{array}{ccccccccc}
 0 & 0 & 0 & 0 & 0 & 0 & 0 & 0 & 0   \\
 0 & 1 & 2 & 0 & 1 & 2 & 0 & 1 & 2  \\
 0 & 2 & 1 & 0 & 2 & 1 & 0 & 2 & 1  \\
 0 & 0 & 0 & 1 & 1 & 1 & 2 & 2 & 2  \\
 0 & 1 & 2 & 1 & 2 & 0 & 2 & 0 & 1  \\
 0 & 2 & 1 & 1 & 0 & 2 & 2 & 1 & 0 \\
 0 & 0 & 0 & 2 & 2 & 2 & 1 & 1 & 1 \\
 0 & 1 & 2 & 2 & 0 & 1 & 1 & 2 & 0 \\
 0 & 2 & 1 & 2 & 1 & 0 & 1 & 0 & 2
\end{array}\right),\quad B^{(1)}=\left(\begin{array}{ccc}
 0 & 0  & 0 \\
  0 & 1 & 2\\
 0 & 2 & 1
\end{array}\right).$$
Assume we receive $w=(1, 1, 1, 1, 1, 1, 1, 1, 0, 1, 0, 0)$.
Then,
\begin{align*}
    S_{w_1} &= ((0,1,0), (0,1,0), (0,1,0), (0,1,0), (0,1,0), (0,1,0), (0,1,0), (0,1,0), (1, 0, 0)),\\
    S_{w_2} &= ((0,1,0), (1, 0, 0), (1, 0, 0)),
\end{align*}
and
$$B_z^{(2)}\star S_{w_1}=\left(\begin{array}{ccc}
  1   & 8 & 0\\
   3  & 4 & 2\\
   2 & 3 & 4\\
   3 & 4 & \mathbf{2}\\
   2 & 3 & 4\\
   4 & 2 & 3\\
   2 & 3 & 4\\
   4 & 2 & 3\\
   3 & 4 & 2
\end{array}\right),\quad B_z^{(1)}\star S_{w_2}=\left(\begin{array}{ccc}
  2   & 1 & 0\\
   \textbf{0}  & 2 & 1\\
   0 & 2 & 1
\end{array}\right).$$
Consider now for example the codeword $v=2\alpha^{(1)}+\alpha^{(3)}$. Then, $N(v,w)=N(v_1,w_1)+N(v_2,w_2)=2+0=2$; see the bold numbers in $B_z^{(2)}\star S_{w_1}$ and $B_z^{(1)}\star S_{w_2}$.
\end{example}

Recall from Theorem \ref{cor:complexity_decoding_RM} that
$B^{(m-i)}_z\star S_{w_i}$ can be calculated using at most $q(q-1)q^{m-i}(m-i)$ additions.

\begin{theorem}
\label{cor:complexity_comparing_received_MacDonald}
To calculate the distances between a received word and all codewords of $\mathcal{C}_{m,m-k}(q)$, at most $$q^2 n\log_q(n)$$ additions are required.
\end{theorem}

\begin{proof}
To calculate the distances between a received word and all codewords of $\mathcal{C}_{m,m-k}(q)$, we split the received word into $k$ parts of suitable sizes and calculate the distances between $R(q,m-i)$ and the $i$-th part of the received word for $i\in\{1,\hdots,k\}$ according to equation \eqref{eq:definition_R}.
This requires to calculate $B^{(m-i)}_z\star S_{w_i}$, for $i\in\{1, \ldots, k \}$.
According to Theorem \ref{cor:complexity_decoding_RM}, this has complexity equal to
$$\sum_{i=1}^kq(q-1)q^{m-i}(m-i)\leq (m-1) q(q-1)\sum_{i=1}^k q^{m-i}\leq\log_q(n) q(q-1) n
   $$
since $\sum_{i=1}^k q^{m-i}=n$ and $q^{m-1}\leq n$. Moreover, for each of the $q^{m}$ codewords, we have to add the corresponding entries
of $B^{(m-i)}_z\star S_{w_i}$ for $i\in\{1, \ldots, k \}$, requiring
$$q^{m}(k-1)\leq q n (k-1)\leq q n \log_q(n)$$additions.
\end{proof}

\subsection{Comparing a received word to all codewords of certain subcodes of MacDonald codes}

\label{sec:complexity_subcodes}

In Section \ref{sec:NewDecodingAlgorithm} we will use the decoding of MacDonald codes in our improved version of the Viterbi algorithm.
In order to deal with the first and last steps of the Viterbi algorithm we also study the decoding of particular classes of subcodes of MacDonald codes.

Note that in each time-step $\mu+1\leq t\leq  N-\mu$ in the Viterbi algorithm one has to compute the distances between $r_{t-1}$ and all the codewords of the block code generated by $(G_0^\top\ \cdots\ \ G_{\mu-1}^\top \ \tilde{G}_{\mu}^\top)^\top$ (if this matrix has more columns than rows such that it is the generator matrix of a block code).
Assuming $N\geq 2\mu+1$, for $t\leq \mu$, one has to consider the matrix $(G_0^\top\ \cdots\ \  G_{t-1}^\top)^\top$ and for $ N-\mu+1\leq t\leq N-1$, the matrix $(G_{t-N+\mu}^\top\ \cdots\ \ G_{\mu-1}^\top \ \tilde{G}_{\mu}^\top)^\top$ and for $t=N$, the matrix $\tilde{G}_{\mu}$.
For
$$\left(\begin{array}{c}
    G_0\\ \vdots \\\tilde{G}_{\mu}
 \end{array}\right)=C_{m,m-k}(q)$$ as in \eqref{eq:definition_R}, i.e., $n=\frac{q^{\delta}(q^k-1)}{q-1}$, one gets
for $0\leq j\leq \mu-1$:
\begin{align*}
 \hspace{-1mm}\left(\begin{array}{c} G_0\\ \vdots\\ G_j\end{array}\right) \hspace{-0.5mm}= \hspace{-0.5mm}\left(
    \begin{array}{ c|c |c|c }
   R(q,(j+1)k-1)_{q^{\delta-jk}} & 0_{1\times q^{\delta+k-2}} & \cdots &   \\
        &   R(q,(j+1)k-2)_{q^{\delta-jk}}& \cdots &   0_{(k-1)\times q^{\delta}}\\
      &    &                \cdots      &
         \\
        &               &   \vdots  &   R(q,jk)_{q^{\delta-jk}}
        \end{array}
    \right)
\end{align*}
where the subscript $q^{\delta-jk}$ indicates that there are $q^{\delta-jk}$ copies of the corresponding Reed-Muller codes.
We write $\mathcal{C}'_{k(j+1)}(q)$ for the (block) code with generator matrix  $(G_0^\top\ \cdots\ \  G_{j}^\top)^\top$.
Considering this generator matrix in the form above can be used to make also the first steps of the Viterbi algorithm more efficient.
The complexity for step $t=j+1\leq \mu$ in the Viterbi decoding, i.e., comparing a received word to all codewords of the block code with generator matrix $(G_0^\top\ \cdots\ \  G_{j}^\top)^\top$, is upper bounded by
\begin{align*}
&q^{\delta-jk} q (q-1) \sum_{i=jk}^{(j+1)k-1}q^i  i+ q^{k(j+1)} k q^{\delta-kj}\\
&\leq ((j+1)k-1) q(q-1) n+k q n\leq q^2 n \log_q(n).
\end{align*}
Note for the second term that we have $q^{k(j+1)}$ codewords and $k q^{\delta-kj}$ copies of first order Reed-Muller codes.
Hence, adding the results for each copy for each codeword requires at most
$q^{k(j+1)} k q^{\delta-kj} \leq k q n \leq q \log_q(n) n$ additions.

Also the last steps of the Viterbi algorithm can be considered in this way, here using the matrices $(G_j^\top\ \cdots\ \  \tilde{G}_{\mu}^\top)^\top$. If we permute the columns of this matrix in a suitable way, we obtain that it consists of several copies of the generator matrix of a Reed-Muller code where the all-one row is deleted, i.e.,
$$\left(\begin{array}{c}
    G_j\\ \vdots\\ \tilde{G}_{\mu}\end{array}\right)=
        \underbrace{\begin{pmatrix}R'(q,\delta-(j-1)k) & \cdots & R'(q,\delta-(j-1)k)\end{pmatrix}}_{\frac{q^{(j-1)k}(q^k-1)}{q-1}\ \text{times}}.
    $$
We denote the (block) code with such a generator matrix by $\mathcal{R}^*(q,\delta-(j-1)k)$.
Since we have $q^{\delta-(j-1)k}$ codewords and $\frac{q^{(j-1)k}(q^k-1)}{q-1}$ copies of $R'(q,\delta-(j-1)k)$, adding the results for each copy for each codeword requires at most $\frac{q^{\delta}(q^k-1)}{q-1}=n$ additions.
Therefore, the overall complexity of comparing all codewords of $\mathcal{R}^*(q,\delta-(j-1)k)$ to a received word is upper bounded by
\begin{align*}
&q(q-1) (\delta-(j-1)k) q^{\delta-(j-1)k}\frac{q^{(j-1)k}(q^k-1)}{q-1}+n\\
&\leq q^2 n \log_q(n).
\end{align*}

\section{Improved Viterbi algorithm for a class of optimal convolutional codes}
\label{sec:NewDecodingAlgorithm}

In this section, we combine the efficient decoding of the block codes of the previous section with the Viterbi algorithm to obtain a more efficient version of the Viterbi algorithm for the convolutional codes of Construction 1.
The resulting algorithm is given in Algorithm \ref{alg:cap}.

\begin{algorithm}
\caption{Improved Viterbi Algorithm}\label{alg:cap}
\flushleft Let $\mathcal{C}$ be as in Construction 1 with minimal generator matrix $G(z)$, memory $\mu$, received word $r = (r_0, \ldots, r_{N-1}) \in \mathbb{F}_q^{n \cdot N}$ to be decoded, input length $N$.\\
\flushleft\textbf{Step 1:}
\begin{algorithmic}
\For {$t$ from $1$ to  $\mu$}\State \multiline{Compute the distances between $r_{t-1}$ and the codewords of $\mathcal{C}'_{kt}(q)$ and
for each of the $q^{tk}$ reached states save its metric $d(r_{[0,t-1]},c_{[0,t-1]})$.
}
\EndFor
\end{algorithmic}
\flushleft\textbf{Step 2:}
\begin{algorithmic}
\For {$t$ from $\mu +1$ to  $N-\mu$}\State \multiline{Compute the distances between $r_{t-1}$ and the codewords of $\mathcal{C}_{\delta+k,\delta}(q)$.
Save the result and
for each state and the corresponding $q^k$ incoming paths do the following:
}
\State \multiline{
1) Add the branch metric $d(r_{t-1},c_{t-1})$ entering the state to the partial path metric of the corresponding survivor at $t-1$.
\\
2) Compare the partial path metrics of all $q^{k}$ paths entering each state.\\
3) For each state, save the path with the smallest partial path metric (called the survivor) and its metric $d(r_{[0,t-1]},c_{[0,t-1]})$.
If there are multiple such paths take one randomly.
Eliminate all other paths in the trellis.
}
\EndFor
\end{algorithmic}
\flushleft\textbf{Step 3:}
\begin{algorithmic}
\For {$t$ from $N-\mu+1$ to $N$}\State \multiline{Compute the distances between $r_{t-1}$ and the codewords of $\mathcal{R}^*(q,\delta-(t-N+\mu-1)k)$.
Save the result and
for each state and the corresponding incoming paths do the following:
}
\State \multiline{
1) Add the branch metric $d(r_{t-1},c_{t-1})$ entering the state to the partial path metric of the corresponding survivor at $t-1$.
\\
2) Compare the partial path metrics of all paths entering each state.\\
3) For each state, save the path with the smallest partial path metric (called the survivor) and its metric $d(r_{[0,t-1]},c_{[0,t-1]})$.
If there are multiple such paths take one randomly.
Eliminate all other paths in the trellis.
}
\EndFor
\end{algorithmic}
\flushleft\textbf{Step 4:} Return the single survivor path and its metric.
\end{algorithm}

Next, we study the complexity of this algorithm and compare it to the classical Viterbi algorithm.
The operations we count are additions of integers and comparisons.
We get the following result.
\begin{theorem}
\label{thm:complexity_improved_viterbi}
The complexity of Algorithm \ref{alg:cap} for codes as in Construction 1 is bounded from above by $2 q^2 n \log_q(n) N$.
In particular it is in $O(N q^2n \log_q(n))$.
\end{theorem}
\begin{proof}
    We bound the complexity at each time unit $t$ and add them up.
    In Step $2$ the complexity of computing the distances between $r_{t-1}$ and the codewords of $\mathcal{C}_{\delta+k,\delta}(q)$ is by Theorem \ref{cor:complexity_comparing_received_MacDonald} bounded by $q^2 n \log_q(n)$.
    In Steps 1 and 3 the complexity of computing the distances is also bounded by $q^2 n \log_q(n)$ according to Section \ref{sec:complexity_subcodes}.
    In all three steps the complexity of comparing the path metrics for each state is bounded by $q^\delta q^k = q^{\delta + k}$.
    Since Construction $1$ satisfies
    \[
        q^2 n \log_q(n) \geq q^{\delta + k},
    \]
    we can bound the number of additions and comparisons needed per time unit $t$ by
    \[
        2 q^2 n \log_q(n).
    \]
    We have $N$ time units $t$ overall and so we get as an upper bound for the complexity
    \[
        2 q^2 n \log_q(n) N
    \]
    which gives the result.
\end{proof}

Recall that we have by definition that $f = \Theta(g)$ if and only if $f = O(g)$ and $g = O(f)$.
Compared to our improved Algorithm \ref{alg:cap}, the Viterbi algorithm has complexity
\begin{align*}
& \Theta\left(\sum_{t=1}^{\mu}q^{tk} n+\sum_{t=\mu+1}^{N-\mu}q^{\delta} q^k n+\sum_{t=N-\mu+1}^Nq^k q^{\delta-(t-N+\mu)k} n\right)\\&=\Theta\left((N-2\mu+2) q^{\delta+k} n\right)= \Theta\left(\left(N-2\left\lceil\frac{\delta}{k}\right\rceil+2\right)  q n^2\right).
\end{align*}
Note that $n = \frac{q^{\delta}(q^k-1)}{q-1}$ and therefore the complexity in Theorem \ref{thm:complexity_improved_viterbi} is significantly smaller.
More precisely, using $n\approx q^{\delta+k-1}$, with Algorithm \ref{alg:cap} we are able to reduce the decoding complexity by a factor of $\approx \frac{\delta+k-1}{q^{\delta+k-2}}$ compared to the classical Viterbi algorithm.

\section{Alternative constructions with good column distances and efficient decoding algorithm}
\label{sec:alternative_constructions}

In this section, we will present two alternative constructions for convolutional codes with good column distances.
These constructions do not yield convolutional codes with optimal column distances but we include them in order to cover code rates that cannot be covered by Construction 1.

Similar to Construction 1, we will use first order Reed-Muller codes as building blocks for the alternative constructions to enable similar efficient decoding with the Viterbi algorithm.

Recall that Construction $1$ can only cover certain values of $n$.
The same will be true for the new constructions.
We will show the following.
If we fix the code parameters $(n,k,\delta)$ as well as the field size $q$ for one of the new constructions and we take Construction $1$ for the same $q$ but initialized with the next smaller value of $n$ that is possible for Construction $1$, then the column distances of the new construction are at least as large.

\subsection{Construction of convolutional codes with good column distances}

In contrast to Construction 1 and the later presented Construction 3, for the following Construction 2, the order of the row degrees matters and for this construction we assume the row degrees to be in increasing order, i.e., $\tilde{G}_{\mu}$ consists of the last rows of $G_{\mu}$.

\begin{theorem}[Construction 2]
\label{thm:construction_2}
    The $(q^{\delta+k-1},k,\delta)$ convolutional code $\mathcal{C}$ over $\mathbb F_q$ with generator matrix $G(z)=\sum_{i=0}^{\mu}G_iz^i$ where $\mu=\left\lceil\frac{\delta}{k}\right\rceil$ and
$$\left(\begin{array}{c}
    G_0\\ \vdots\\
    \tilde{G}_{\mu}
    \end{array}\right)=
   R(q,\delta+k-1),$$
where we assume the first row to be the all-one row, is non-catastrophic and has for $k > 1$ column distances equal to
$$d_j^c(\mathcal{C})=\begin{cases}
(j+1)n\frac{q-1}{q} & \text{for}\quad j\leq\left\lfloor\frac{\delta}{k}\right\rfloor\\
\left(\left\lfloor\frac{\delta}{k}\right\rfloor+1\right)n\frac{q-1}{q}  & \text{for}\quad j\geq\left\lfloor\frac{\delta}{k}\right\rfloor
\end{cases} \quad \text{if}\ \delta\not\equiv k-1\mod k,$$
and
$$d_j^c(\mathcal{C})=\begin{cases}
(j+1)n\frac{q-1}{q} & \text{for}\quad j<\left\lfloor\frac{\delta}{k}\right\rfloor+1=\mu\\
n\left(1+\left\lfloor\frac{\delta}{k}\right\rfloor \frac{q-1}{q}\right)  & \text{for}\quad j\geq\left\lfloor\frac{\delta}{k}\right\rfloor+1=\mu
\end{cases}
\quad\text{if}\
\delta\equiv k-1\mod k
.$$
For $k = 1$ this construction coincides with Construction $1$.
\end{theorem}

\begin{proof}
Similar to the proof for Construction 1, the convolutional codes of Construction 2 are non-catastrophic because $R(q, \delta + k - 1)$ is a submatrix of
$$\left(\begin{array}{c}
    G_0\\ \vdots\\
    \tilde{G}_{\mu}
    \end{array}\right).$$
Here these two matrices are actually identical.

For the calculation of the column distances, we can assume $k>1$ because otherwise Construction 2 is identical with Construction 1.
Then $G_0$ has at least two rows and hence, does not just consist of the all-one row.
Therefore, we have $d_0^c(\mathcal{C})=n\frac{q-1}{q}$.
Moreover, for $j>0$, similar to Construction 1 for $k=1$, because of
$$\operatorname{wt}\left((u_0,\hdots,u_j)\left(\begin{array}{c}
    G_j\\ \vdots\\ G_0
\end{array}\right)\right)=n \frac{q-1}{q}$$ for all $(u_0,\hdots,u_j)$ with $u_0\neq 0$, one gets
\begin{align*}
d_j^c(\mathcal{C})=(j+1)n\frac{q-1}{q}\quad\text{for}\quad j\leq\left\lfloor\frac{\delta}{k}\right\rfloor.
\end{align*}
For $j>\left\lfloor\frac{\delta}{k}\right\rfloor$ we have to distinguish two cases.

If $\delta\not\equiv k-1\mod k$ and $k\mid\delta$, it is clear that the column distances cannot further increase for $j\geq\frac{\delta}{k}=\mu$.
If $\delta\not\equiv k-1\mod k$ and $k\nmid\delta$, then for $j \geq \mu = \left\lfloor \delta/k \right\rfloor + 1$ at least the first two rows of $G_{j}$ are zero and the weight of the second row of
$$G_j^c=\left(\begin{array}{ccc}
    G_0 & \cdots & G_j\\
    & \ddots & \vdots\\
    & & G_0
\end{array}\right)$$
is equal to $$\left(\left\lfloor\frac{\delta}{k}\right\rfloor+1\right)n \frac {q-1}{q},$$ and therefore the column distances cannot increase any further than this value.

If $\delta\equiv k-1\mod k$, then only the first row of $G_{\mu}$ is zero. In this case, the weight of the first row of $G_{\mu}^c$ is equal to
$$n+(\mu-1)n\frac{q-1}{q}=n\left(1+\left\lfloor\frac{\delta}{k}\right\rfloor \frac{q-1}{q}\right).$$
The weight of any other linear combination of rows of $G_{\mu}^c$ involving at least one of the first $k$ rows has
weight
$$(\mu+1)n\frac{q-1}{q}=\left(\left\lfloor\frac{\delta}{k}\right\rfloor+2\right)n \frac{q-1}{q},$$
since all $\mu+1$ parts of length $n$ have weight $n\frac{q-1}{q}$ by Lemma \ref{lemma:weight_Mac_Donald_Codes}.
In sum, one obtains $$d_{\mu}^c(\mathcal{C})=n\left(1+\left\lfloor\frac{\delta}{k}\right\rfloor \frac{q-1}{q}\right).$$
Since $G_j=0$ for $j>\mu$, the column distances cannot increase any further.
\end{proof}

Since for $k=1$, Construction 2 and Construction 1 coincide, Example \ref{ex:k1} provides an example in which the previous Construction 2 gives a convolutional code with optimal column distances. In the following, we will provide another example for Construction 2 that leads to optimal column distances but differs from Construction 1.

\begin{example}
As in Example \ref{ex:macDonald}, we choose $q=3$, $\delta=1$ and $k=2$, i.e., $\delta\equiv k-1\mod k$. One obtains
$$R(3,2)=\left(\begin{array}{ccccccccc}
      1 & 1 & 1 & 1 & 1 & 1 & 1 & 1 & 1  \\
          0 & 1 & 2   & 0 & 1 & 2 &  0 & 1 & 2   \\
     \hline
    0 & 0 &  0   & 1 & 1 & 1  & 2 & 2 & 2
  \end{array}\right)=\begin{pmatrix}
      G_0\\ \tilde{G}_1
  \end{pmatrix},$$
  i.e., $G(z)=G_0+G_1z$ with
\begin{align*}
G_0&=\left(\begin{array}{ccccccccc}
      1 & 1 & 1 &  1 & 1 & 1   & 1 & 1 & 1  \\
      0 & 1 & 2   & 0 & 1 & 2   & 0 & 1 & 2
  \end{array}\right)\\
   G_1&=\left(\begin{array}{ccccccccc}
   0 & 0 & 0 & 0 & 0 & 0 & 0 &0 & 0 \\
0 & 0 & 0  &  1 & 1 & 1  & 2 & 2 & 2
  \end{array}\right)
\end{align*}
is the generator matrix of a $(9,2,1)$ convolutional code over $\mathbb F_3$ with column distances
$d_0^c=6$ and $d_j^c=9$ for $j\geq 1$.
Elementary calculations show that these column distances are actually optimal for the given field and code parameters.
If one changed the order of the rows of $G_1$, then the resulting convolutional code would have smaller column distances equal to $d_j^c=6$ for all $j\in\mathbb{Z}_{\geq 0}$.
The order of the rows matters here, since neither $G_0$ nor $G_{\mu}=G_1$ generates a one-weight code.

In comparison with this a $(9,2,1)$ convolutional code is MDP if and only if $d_0^c=8$ (since in Definition \ref{def:mdp} we have $L=0$ for these parameters), i.e., $G_0$ has to be the generator matrix of an MDS block code and consequently $q\geq 8$ is required from the MDS conjecture for $k=2$.
\end{example}

Finally, we give an example, where the column distances of Construction 2 are not optimal.

\begin{example}
Consider the case $k=\delta=q=2$. Then,
$$R(2,3)=\begin{pmatrix}
    1 & 1 & 1 & 1 & 1 & 1 & 1 & 1\\
    1 & 1 & 1 & 0 & 1 & 0 & 0 & 0\\
    \hline
    1 & 1 & 0 & 1 & 0 & 1 & 0 & 0\\
    1 & 0 & 1 & 1 & 0 & 0 & 1 & 0
\end{pmatrix}=\begin{pmatrix}
    G_0\\ G_1
\end{pmatrix},$$
leading to an $(8,2,2)$ convolutional code with $d^c_0=4$ and $d^c_j=8$ for $j\geq 1$.
One example of a binary $[8,2]$ block code with largest possible minimum distance for the given parameters is the code with generator matrix
$$\hat{G}_0=\begin{pmatrix}
    1 & 1 & 0 & 1 & 1 & 0 & 1 & 1\\
    1 & 0 & 1 & 1 & 0 & 1 & 0 & 1
\end{pmatrix}.$$
Its minimum distance is $5$, which is optimal since to achieve minimum distance $\geq 6$ both rows of the generator matrix would need to have weight $\geq 6$ implying that their sum has weight $\leq 4$.
Therefore, $d^c_0=4$ is not optimal but close to optimal.
\end{example}

\begin{theorem}[Construction 3]
\label{thm:construction_3}
    The $(\frac{q^{\delta+k}-1}{q-1},k,\delta)$ convolutional code $\mathcal{C}$ over $\mathbb F_q$ with generator matrix $G(z)=\sum_{i=0}^{\mu}G_iz^i$ where $\mu=\left\lceil\frac{\delta}{k}\right\rceil$ and
$$\left(\begin{array}{c}
    G_0\\ \vdots\\
    \tilde{G}_{\mu}
    \end{array}\right)=
   S(q,\delta+k)$$
with $m:=\delta+k$ is non-catastrophic and has column distances equal to
$$d_j^c(\mathcal{C})=\begin{cases}
(j+1) n\cdot\frac{q^{m-1}(q-1)}{q^m-1} & \text{for}\quad j\leq\left\lfloor\frac{\delta}{k}\right\rfloor\\
\left(\left\lfloor\frac{\delta}{k}\right\rfloor+1\right)n\cdot\frac{q^{m-1}(q-1)}{q^m-1} & \text{for}\quad j\geq\left\lfloor\frac{\delta}{k}\right\rfloor
\end{cases}.$$
Moreover,
$$d_{\mathrm{free}}(\mathcal{C})=\left(\left\lfloor\frac{\delta}{k}\right\rfloor+1\right)n\cdot\frac{q^{m-1}(q-1)}{q^m-1}.$$
\end{theorem}

\begin{proof}
Similar to the proof for Construction 1, the convolutional codes of Construction 3 are non-catastrophic because $R(q,m-1)$ is a submatrix of
$$\left(\begin{array}{c}
    G_0\\ \vdots\\
    \tilde{G}_{\mu}
    \end{array}\right).$$
The formula for the column distances follows from the fact that all non-zero codewords of $\mathcal{S}(q,\delta+k)$ have weight $q^{\delta+k-1}$.
Since $\mathcal{C}$ is non-catastrophic, the formula for the free distance follows directly from $d_{\mathrm{free}}(\mathcal{C})=\lim_{j\rightarrow\infty}d_j^c(\mathcal{C})$.
\end{proof}

\begin{example}
Consider again the case $q=3$, $k=2$, $\delta=1$ as in Example \ref{ex:simplex}.
One obtains the $(13,2,1)$ convolutional code over $\mathbb F_3$ given by
\begin{align*}
 G_0&=\left(\begin{array}{ccccccccccccc}
      1 & 1 & 1 & 0  & 1 & 1 & 1 & 0  & 1 & 1 & 1 & 0 & 0\\
      0 & 1 & 2 & 1 &  0 & 1 & 2 & 1 &  0 & 1 & 2 & 1  & 0
  \end{array}\right)\\
 G_1&=\left(\begin{array}{ccccccccccccc}0 & 0 & 0 & 0  & 1 & 1 & 1 & 1 &  2 & 2 & 2 & 2  & 1\\ 0 & 0 & 0 & 0 & 0 & 0 & 0 & 0 & 0 & 0 & 0 & 0 & 0
  \end{array}\right)
\end{align*}
with $d^c_j=9$ for all $j\in\mathbb{Z}_{\geq 0}$.

Now consider what happens if we change the last column of $G_0$ to $\begin{pmatrix}
    0\\1
\end{pmatrix}$ or $\begin{pmatrix}
    1\\1
\end{pmatrix}$
to obtain a matrix $\hat{G_0}$
and set $\hat{G}_1 = G_1$.
Note that the puncturing of the codes generated by $\hat{G_0}$ and $G_0$ in the last coordinate is a one-weight code with minimum distance $9$.
Since there exists a linear combination of the rows of $\hat{G}_0$ such that the last coordinate is equal to $0$, the block code generated by $\hat{G}_0$ has also minimum distance $9$.
However, the second row of $\hat{G}_0$, which forms together with the second row of $\hat{G}_1$ the second row of $\hat{G}_1^c$, has then weight $10$, implying upon inspection of $\hat{{G}_1^c}$ that $d_1^c$ increases to $10$. Therefore, Construction 3 is not ODP for these parameters.
\end{example}

\subsection{Improved Viterbi decoding for Construction 2 and Construction 3}

Moreover, as we will show in the following, also the codes of Constructions 2 and 3 can be efficiently decoded with the Viterbi algorithm.

\begin{theorem}
    The code from Construction $2$ can be decoded with a modified Viterbi algorithm with complexity upper bounded by $O(N q^2 n \log_q(n))$.
\end{theorem}
\begin{proof}
    For Construction 2, for $t\in\{\mu+1,\hdots,N-\mu\}$, we can directly decode using the first order Reed-Muller code $\mathcal{R}(q,\delta+k-1)$.
    For $t\in\{1,\hdots,\mu\}$, we have to consider $q^{\delta-(t-1)k}$ copies of $\mathcal{R}(q,tk-1)$, and for $t\in\{N-\mu+1,\hdots, N\}$, we have to consider $q^{(t-N+\mu)k-1}$ copies of $\mathcal{R}'(q,\delta-(t-N+\mu-1)k)$.
    Using the fact that comparing a received word to all codewords of $\mathcal{R}(q, \nu)$ can be done with complexity $q(q-1)q^\nu \nu$, we can calculate that for each $t\in\{1,\hdots, N\}$, the complexity of each step $t$ in the Viterbi algorithm is $O(q^2n\log_q(n))$, which yields the same overall decoding complexity as for Construction 1.
    Note that in the cases $t\in\{1,\hdots,\mu\}$ and $t\in\{N-\mu+1,\hdots, N\}$ where we have several copies of  a code to consider, adding the results of the single copies for all codewords requires at most $q n$ additions.
    This is the case because if the code we consider consists of copies of $\mathcal{R}'(q,d)$ for some $d\in\mathbb{Z}_{\geq 1}$, then there are $\frac{n}{q^d}$ copies and each copy has $q^d$ codewords and if the code we consider consists of copies of $\mathcal{R}(q,d)$ for some $d\in\mathbb{Z}_{\geq 1}$, then there are $\frac{n}{q^d}$ copies and each copy has $q^{d+1}$ codewords.
\end{proof}

\begin{theorem}
    The code from Construction $3$ can be decoded with a modified Viterbi algorithm with complexity upper bounded by $O(N q^2 n \log_q(n))$.
\end{theorem}
\begin{proof}
    Considering Construction 3,
    we can write
    \[
        S(q,m)=\left( R(q,m-1) \ \ \begin{array}{c}
        0_{1 \times \frac{q^{m-1}-1}{q-1}}\\
      S(q,m-1)
    \end{array}\right)
    \]
    and consequently,
    \small
    $$\left(\begin{array}{c}
    G_0\\ \vdots\\
    \tilde{G}_{\mu}
    \end{array}\right)=\left(
      \begin{array}{c|c|c|c|c|c}
            R(q,m-1)   & 0_{1 \times q^{m-2}}            & 0_{2 \times q^{m-3}} & \ldots & 0_{(m-2) \times q^{1}} & 0_{(m-1) \times 1}\\
                                   & R(q,m-2)  & & \ldots & & \\
                                   &                       & R(q,m-3) &\ldots &   & \\
                              &                 &    &\vdots     &  & \\
                                    &                       &           &        & R(q,1) & \\
            & & & & & 1
        \end{array}
        \right).$$
        \normalsize
    Hence, comparing all codewords of a simplex code to a received word can be done with at most
    \begin{align*}
      &\sum_{i=1}^{m-1}q(q-1)q^{m-i}(m-i)+q^m (m-1)\leq (m-1) q(q-1)\sum_{i=1}^{m-1} q^{i}+q^m (m-1)\\
      &\leq q^2 \log_q(n) n
    \end{align*}
    additions.
    Here the term $q^m  (m-1)$ comes from summing up the results from the single Reed-Muller codes.

    Similar to Construction 1, considering steps with $t\in\{1,\hdots,\mu\}$ of the Viterbi algorithm, for $0\leq j=t-1\leq \mu-1$, requires comparing all codewords of the code with generator matrix
    \begin{align*}&\left(\begin{array}{c} G_0\\ \vdots\\ G_j\end{array}\right)
    \end{align*}
    to a received word. This generator matrix is up to column permutations equal to
    \small
    \begin{align*}
    &\left(
    \setlength\arraycolsep{1pt}
        \begin{array}{ c|c |c|c|c }
      R(q,(j+1)k-1)_{q^{\delta-jk}} & 0_{1\times q^{m-2}} & \cdots &   \\
            & R(q,(j+1)k-2)_{q^{\delta-jk}} & \cdots &   0_{((j+1)k-1)\times q^{\delta-kj}}\\
          &    &                \cdots &     & 0_{(j+1)k\times \frac{q^{\delta-kj}-1}{q-1}}
             \\
            &              &   \vdots  & R(q,0)_{q^{\delta-jk}}
            \end{array}
        \right)
    \end{align*}
    \normalsize
    where $R(q,d)_{q^{\delta-jk}}$ denotes the matrix consisting of $q^{m-k(j+1)}=q^{\delta-jk}$ copies of $R(q,d)$.
    Therefore, at most
    \begin{align*}
    &q^{\delta-jk} q (q-1) \sum_{i=1}^{(j+1)k-1}q^i  i+ q^{k(j+1)} (j+1)k q^{m-k(j+1)}\\
    &\leq q(q-1)((j+1)k-1) q^{\delta-jk} \frac{q^{(j+1)k}-1}{q-1}+(j+1)k q q^{m-1}\\
    &\leq ((j+1)k-1) q(q-1) n+(j+1)k q n\leq q^2 n\log_q(n)
    \end{align*}
    additions are required, where the first part of the sum gives the required number of additions for calculating the distances between all codewords of each of the involved Reed-Muller codes to the corresponding part of the received word and the second part gives for each of the $q^{k(j+1)}$ codewords the required number of additions to sum up the results for the single Reed-Muller codes.

    Finally, we consider steps with $t\in\{N-\mu+1,\hdots, N\}$ of the Viterbi algorithm. For $1\leq j\leq \mu$, we compare a received word to all codewords of the code with generator matrix
    $$\left(\begin{array}{c}
        G_j\\ \vdots\\ \tilde{G}_{\mu}\end{array}\right).
    $$
    This matrix is up to column permutations equal to
    \small
    $$
    \setlength\arraycolsep{2.8pt}
    \begin{pmatrix}
               \hspace{-0.4mm}R'(q, \delta \hspace{-0.4mm} - \hspace{-0.4mm} (j \hspace{-0.4mm} - \hspace{-0.4mm} 1)k)_{\frac{q^{jk}-1}{q-1}}& R(q,\delta \hspace{-0.4mm} - \hspace{-0.4mm} (j \hspace{-0.4mm} - \hspace{-0.4mm} 1)k \hspace{-0.4mm} - \hspace{-0.4mm} 1) & 0_{1\times q^{\delta-(j-1)k-2}} & \cdots & 0_{(\delta - (j - 1)k  - 1)\times 1}\\ &  & R(q,\delta \hspace{-0.4mm} - \hspace{-0.4mm} ( j \hspace{-0.4mm} - \hspace{-0.4mm}1)k \hspace{-0.4mm} - \hspace{-0.4mm} 2) & \cdots & 1 \end{pmatrix}\hspace{-1mm},$$
    \normalsize
    where $R'(q, \delta-(j-1)k)_\frac{q^{jk}-1}{q-1}$ denotes $\frac{q^{jk}-1}{q-1}$ copies of $R'(q, \delta-(j-1)k)$.
    The comparison requires at most
    \begin{align*}
    &q(q-1) \left(\underbrace{\frac{q^{jk}-1}{q-1} q^{\delta-(j-1)k}(\delta-(j-1)k)}_{\frac{q^{jk}-1}{q-1}\ \text{times}\ R'(q,\delta-(j-1)k)}+\sum_{i=1}^{\delta-(j-1)k-1}q^i i \right)\\
    &\quad+ q^{\delta-(j-1)k}\left(\frac{q^{jk}-1}{q-1}+\delta-(j-1)k-1\right)\\
    &\leq q(q-1) \left(\frac{q^{jk}-1}{q-1} q^{\delta-(j-1)k}(\delta-(j-1)k)+\sum_{i=1}^{\delta-(j-1)k-1}q^i i \right)+ n + n \log_q(n)\\
    &\leq q(q-1) (\delta-(j-1)k)\left(\frac{q^{\delta+k}-q^{\delta-(j-1)k}}{q-1}+\frac{q^{\delta-(j-1)k}-1}{q-1}\right)+n+ n\log_q(n)\\
    &\leq q(q-1) \log_q(n) n+n+ n\log_q(n)\\
    &\leq q^2 n\log_q(n)
    \end{align*}
    additions.
\end{proof}

In summary, Constructions 1, 2 and 3 can all be decoded with the same asymptotic complexity, which is significantly smaller than the complexity for general Viterbi decoding.

\subsection{Comparison of constructions}
To compare the three constructions of this paper, we denote by $d_{j,\ell}$ the $j$-th column distance of Construction $\ell$ and by $n_\ell$ its length. With this notation one obtains from Theorems \ref{thm:construction_1}, \ref{thm:construction_2} and \ref{thm:construction_3}
that for  $j\leq\left\lfloor\frac{\delta}{k}\right\rfloor$,
$$
\frac{d_{j,1}}{n_1}=j\frac{q-1}{q}+\frac{q^{k-1}(q-1)}{q^k-1},\quad
\frac{d_{j,2}}{n_2}=(j+1)\frac{q-1}{q},\quad \frac{d_{j,3}}{n_3}=(j+1)\frac{q-1}{q}\frac{q^{\delta+k}}{q^{\delta+k}-1}
$$
Comparing these quantities, one gets:
\begin{align}\label{comp}
&\frac{d_{j,1}}{n_1}> \frac{d_{j,3}}{n_3}>\frac{d_{j,2}}{n_2}
\end{align}
where the last inequality is obvious.
The following calculations show
$\frac{d_{j,1}}{n_1}> \frac{d_{j,3}}{n_3}$.
\begin{align*}
&j\frac{q-1}{q}+\frac{q^{k-1}(q-1)}{q^k-1}>(j+1)\frac{q-1}{q}\frac{q^{\delta+k}}{q^{\delta+k}-1}\\
&\Leftrightarrow j(q^k-1)(q^{\delta+k}-1)+q^{k}(q^{\delta+k}-1)>(j+1)q^{\delta+k}(q^k-1)\\
&\Leftrightarrow -j(q^{k}-1)-q^k>-q^{\delta+k}\\
&\Leftrightarrow j<\frac{q^{\delta+k}-q^k}{q^k-1}=\frac{q^k(q^{\delta}-1)}{q^k-1},
\end{align*}
which is true for $j\leq \mu\leq\delta$.

The inequalities in \eqref{comp} are actually true for all $j\in\mathbb{Z}_{\geq 0}$ if $\delta\not\equiv k-1\mod k$ and for $j<\mu$ if $\delta\equiv k-1\mod k$.
For the exceptional case $j\geq \mu$ and $\delta\equiv k-1\mod k$, it depends on the concrete code parameters how Construction 2 performs with respect to the other constructions.

The inequalities in \eqref{comp} show that the relative error-correcting capability of Construction 1 is largest, which corresponds to the fact that Construction 1 has optimal column distances.
However, for the same values of $\delta$ and $k$, one has $d_{0,3}=d_{0,1}$ and $d_{j,3}\geq d_{j,1}$ for $j\geq 1$, i.e., the absolute error-correcting capability of Construction 3 is better than the one of Construction 1.
This is possible since $n_3\geq n_1$, that is, Construction 3 has smaller rate than Construction 1.
Also note that Construction 2 has the highest rate among these constructions.

We summarize the comparison of the three presented constructions for $j\leq\lfloor\frac{\delta}{k}\rfloor$ in the following table.
\vspace{-2mm}
   \begin{center}
   \bgroup
\def\arraystretch{1.8}
\resizebox{\textwidth}{!}{
    \begin{tabular}{c|c|c|c|c}
       & $n$ & $d_j^c$ & $d^c_j/n$ & \text{ Complexity} \\
       \hline
       C1 &   $\frac{q^{\delta}(q^k-1)}{q-1}$ & $q^{\delta+k-1}+j(q^{\delta+k-1}-q^{\delta-1})$  & $j\frac{q-1}{q}+\frac{q^{k-1}(q-1)}{q^k-1}$   & $O(N q^2n \log_q(n))$\\
       \hline
       C2 &   $q^{\delta+k-1}$  & $(j+1)q^{\delta+k-2}(q-1)$  &  $(j+1)\frac{q-1}{q}$ & $O(N q^2n \log_q(n))$\\
       \hline
       C3 &   $\frac{q^{\delta+k}-1}{q-1}$  & $(j+1)q^{\delta+k-1}$  & $(j+1)\frac{q-1}{q}\frac{q^{\delta+k}}{q^{\delta+k}-1}$  & $O(N q^2n \log_q(n))$
    \end{tabular}}
\egroup
\end{center}
\vspace{3mm}

The following theorem shows that even if Construction 2 and Construction 3 do not in general have optimal column distances,
these constructions still have better column distances than Construction 1 has for the next smaller value of $n$ that it can cover.

\begin{theorem}
For fixed $k\in\mathbb Z_{\geq 1}$, $j\in\mathbb{Z}_{\geq 0}$, $\ell\in\{2,3\}$ and $n_\ell\geq n_1$, one has $d_{j,\ell}\geq d_{j,1}$.
\end{theorem}

\begin{proof}
We first compare Construction 2 with Construction 1. For the same fixed values of $\delta$ and $k>1$, one has $$n_2=q^{\delta+k-1}<\frac{q^{\delta+k}-q^{\delta}}{q-1}=n_1.$$
To obtain the next smaller possible value for $n_1$, we replace $\delta$ by $\delta-1$ to obtain
$$\frac{q^{\delta+k-1}-q^{\delta-1}}{q-1}<n_2.$$
The $j$-th column distance for Construction 1 with parameters $(\frac{q^{\delta+k-1}-q^{\delta-1}}{q-1}, k,\delta-1)$ is given by $$q^{\delta+k-2}+j(q^{\delta+k-2}-q^{\delta-2})=(j+1)q^{\delta+k-2}-jq^{\delta-2},$$
which is upper bounded by the $j$-th column distance for Construction 2 with parameters $(q^{\delta+k-1},k,\delta)$ given by $(j+1)q^{\delta+k-2}(q-1)$
(and the latter is even larger in the exceptional case mentioned above).

Now we compare Construction 3 with Construction 1.
For the same fixed values of $\delta$ and $k$, one has $$n_3=\frac{q^{\delta+k}-1}{q-1}\geq \frac{q^{\delta+k}-q^{\delta}}{q-1}=n_1.$$
As already mentioned before this theorem, we have $$d_{j,3}=(j+1)q^{\delta+k-1}\geq q^{\delta+k-1}+j(q^{\delta+k-1}-q^{\delta-1})=d_{j,1},$$
which completes the proof.
\end{proof}
In summary, the column distances of Construction $2$ and Construction $3$ are very close to optimal and including these constructions we can cover a larger variety of different code rates for efficiently decodable convolutional codes.

However, all three presented constructions have low rate.
To increase the rate, one can apply puncturing to the presented constructions.
Note that Construction 1 is a puncturing of Construction 3 and Construction 2 is a puncturing of Construction 1.
Because of this and since Construction 1 is the construction with optimal column distances, in the following, we only consider puncturings of Construction 1 and Construction 2.

In general it is a difficult open problem to find the puncturing that one would need to apply to a construction with optimal column distances in order to obtain an ODP convolutional code with higher rate.
Already for optimal $0$-th column distance, i.e., the case of block codes, this is in general an open problem.
However, it is easy to see that when puncturing a convolutional code in one position, the value of the $j$-th column distance can decrease at most by $j+1$.
This implies that if we want to construct an $(n,k,\delta)$ convolutional code with good column distances over $\mathbb F_q$ with $n<\frac{q^{\delta}(q^k-1)}{q-1}$,
one can puncture the generator matrix of Construction 1 in $\frac{q^{\delta}(q^k-1)}{q-1}-n$ positions and, for $j\leq\lfloor\frac{\delta}{k}\rfloor$, one will obtain a convolutional code with
\begin{align}\label{eq:lowerbound}
d_j^c&\geq q^{\delta+k-1}+j(q^{\delta+k-1}-q^{\delta-1})-\left(\frac{q^{\delta}(q^k-1)}{q-1}-n\right)(j+1)\nonumber\\
&= \left(q^{\delta+k-1}- \left(\frac{q^{\delta}(q^k-1)}{q-1}-n\right) \right) + j\left(q^{\delta+k-1}-q^{\delta-1} - \left(\frac{q^{\delta}(q^k-1)}{q-1}-n\right) \right) \nonumber\\
&= \left(n-\frac{q^{\delta+k-1}-q^{\delta}}{q-1}\right) + j \left(n - \frac{q^{\delta}(q^k-1)}{q-1} + \frac{q^{\delta+k-1}(q-1)}{q-1} - \frac{q^{\delta-1}(q-1)}{q-1} \right) \nonumber \\
&=\left(n-\frac{q^{\delta+k-1}-q^{\delta}}{q-1}\right)+j\left(n-\frac{q^{\delta+k-1}-q^{\delta-1}}{q-1}\right).
\end{align}
Note that this lower bound has a positive value if $n\geq\frac{q^{\delta+k-1}-q^{\delta-1}}{q-1}$, where $\frac{q^{\delta+k-1}-q^{\delta-1}}{q-1}$ is exactly the length of Construction 1 if we choose degree $\delta-1$ instead of $\delta$.
This implies that to obtain good constructions for rates that are so large that \eqref{eq:lowerbound} does not give something useful, we can consider puncturings of Construction 1 for smaller values of $\delta$.

Moreover, the lower bound in \eqref{eq:lowerbound} can be improved considering Construction 2 in case $k>1$.
If we plug in $n=q^{\delta+k-1}$ in \eqref{eq:lowerbound}, we obtain for $j=0$,
$$d_0^c\geq \frac{q^{\delta+k}-2q^{\delta+k-1}+q^{\delta}}{q-1}$$
whereas the $0$-th column distance of Construction 2 is $$d_{0,2}=q^{\delta+k-2}(q-1)\geq \frac{q^{\delta+k}-2q^{\delta+k-1}+q^{\delta}}{q-1}$$
for $k\geq 2$ with equality for $k=2$.
This implies that if we puncture Construction 1 such that all columns of Construction 2 are kept as long as $n\geq q^{\delta+k-1}$, we obtain, for $n\geq q^{\delta+k-1}$, the following improved lower bounds for the resulting column distances with $j\leq\lfloor\frac{\delta}{k}\rfloor$:
\begin{align*}
d_j^c&\geq \max\left(\left(n-\frac{q^{\delta+k-1}-q^{\delta}}{q-1}\right)+j\left(n-\frac{q^{\delta+k-1}-q^{\delta-1}}{q-1}\right), (j+1)q^{\delta+k-2}(q-1) \right).
\end{align*}
Besides, for $n<q^{\delta+k-1}$, one obtains with arbitrary puncturings of Construction 2 for $j\leq\left\lfloor\frac{\delta}{k}\right\rfloor$:
\begin{align*}
 d_j^c\geq (j+1)q^{\delta+k-2}(q-1)-(q^{\delta+k-1}-n)(j+1)=(j+1)(n-q^{\delta+k-2}).
\end{align*}
Note that also here this lower bound has a positive value as long as $n>q^{\delta+k-2}$, where $q^{\delta+k-2}$ is the length of Construction 2 if we choose $\delta-1$ as degree instead of $\delta$.
Hence, to obtain codes of even larger rate, one can consider puncturings of Construction 2 with smaller values for $\delta$.

\section{Conclusion and open problems}
\label{sec:conclusion}
In this paper, we present a construction of convolutional codes with optimal column distances over arbitrary finite fields and for arbitrary code parameters $k$ and $\delta$, where $n$ is a function of the other code parameters.
We are able to significantly improve the complexity when decoding them with the Viterbi algorithm.
It remains a challenging open problem to construct convolutional codes with optimal column distances for large code rates.
In particular, it would be interesting to understand how the parameters change after puncturing or shortening in this case or even in the case of block codes.

\section*{Acknowledgements}
We would like to thank Zita Abreu for many helpful discussions.
In addition, we would like to thank the reviewers for their valuable feedback.
The first author is supported by the German research foundation, project number 513811367.
The second author is supported by the Swiss National Science Foundation under SNSF grant number 212865.

\bibliographystyle{abbrv}
\bibliography{literature}

\end{document}